\documentclass[aps,pra,twocolumn,showpacs,superscriptaddress, nofootinbib]{revtex4-2} 
\pdfoutput=1

\usepackage[utf8]{inputenc}
\usepackage[T1]{fontenc}
\usepackage{graphicx}

\usepackage{physics}
\usepackage{bbold}
\usepackage{amsmath, amsthm, amssymb,mathrsfs}

\newtheorem{theorem}{Theorem}

\usepackage{mathtools}

\usepackage[normalem]{ulem}
\usepackage{optidef}
\usepackage{algorithmic}
\usepackage{algorithm}

\usepackage{dsfont}
\usepackage{comment}
\usepackage{xcolor}
\usepackage{vwcol}
\usepackage[colorlinks]{hyperref}

\usepackage[acronym]{glossaries}
\makeglossaries
\newacronym{qkd}{QKD}{quantum key distribution}
\newacronym{di}{DI}{device-independent}
\newacronym{pm}{PM}{prepare-and-measure}
\newacronym{sdp}{SDP}{semidefinite programming}
\newacronym{POVM}{POVM}{Positive Operator Valued Measure}
\newacronym{usd}{USD}{Unambiguous state discrimination}
\newacronym{qber}{QBER}{Quantum Bit Error Rate}

\usepackage{float}
\usepackage[caption=false]{subfig}

\usepackage{tikz}
\usetikzlibrary{
  arrows.meta,
  automata
}
\tikzset{
    -Latex,auto,node distance =1 cm and 1 cm,semithick,
    state/.style ={ellipse, draw, minimum width = 0.7 cm},
    point/.style = {circle, draw, inner sep=0.04cm,fill,node contents={}},
    bidirected/.style={Latex-Latex,dashed},
    el/.style = {inner sep=2pt, align=left, sloped}
}

\usepackage{multirow}
\usepackage{diagbox}

\usepackage{xcolor}

\usepackage[normalem]{ulem}
\newcommand{\stkout}[1]{\ifmmode\text{\sout{\ensuremath{#1}}}\else\sout{#1}\fi}

\begin{document}

\title{Towards a minimal example of quantum nonlocality without inputs}

\author{Sadra Boreiri}
\thanks{These authors contributed equally to this work}
\affiliation{Department of Applied Physics University of Geneva, 1211 Geneva, Switzerland}
\author{Antoine Girardin}
\thanks{These authors contributed equally to this work}
\affiliation{Department of Applied Physics University of Geneva, 1211 Geneva, Switzerland}
\author{Bora Ulu}
\affiliation{Department of Applied Physics University of Geneva, 1211 Geneva, Switzerland}
\author{Patryk Lipka-Bartosik}
\affiliation{Department of Applied Physics University of Geneva, 1211 Geneva, Switzerland}
\author{Nicolas Brunner}
\affiliation{Department of Applied Physics University of Geneva, 1211 Geneva, Switzerland}
\author{Pavel Sekatski}
\affiliation{Department of Applied Physics University of Geneva, 1211 Geneva, Switzerland}

\begin{abstract}
The network scenario offers interesting new perspectives on the phenomenon of quantum nonlocality. Notably, when considering networks with independent sources, it is possible to demonstrate quantum nonlocality without the need for measurements inputs, i.e. with all parties performing a fixed quantum measurement. Here we aim to find minimal examples of this effect. Focusing on the minimal case of the triangle network, we present examples involving output cardinalities of $3-3-3$ and $3-3-2$. A key element is a rigidity result for the Parity Token Counting distribution, which represents a minimal example of rigidity for a classical distribution. Finally, we discuss the prospects of finding an example of quantum nonlocality in the triangle network with binary outputs and point out a connection to the Lovasz local lemma. 
\end{abstract}

\maketitle

\section{Introduction}
The exploration of quantum nonlocality in networks has attracted growing attention in recent years; see e.g. \cite{Tavakoli_Review} for a recent review. This avenue of research opens interesting new perspectives and possibilities for quantum nonlocal correlations.

While quantum nonlocality has been investigated in a broad range of scenarios, including the multipartite case, the network scenario brings a main conceptual novelty. Specifically, the central idea is to consider networks where several sources distribute quantum resources to various subsets of the parties (nodes). In this sense, this model differs from the standard approach to multipartite nonlocality, where all parties are connected via a common source; see e.g. \cite{review}. The main assumption in the network scenario is then to consider that each source in the network is independent of all others \cite{Branciard_2010,Branciard_2012,Fritz_2012}. From a more formal point of view, this implies that the relevant sets of possible correlations are non-convex, as mixing two arbitrary strategies requires a source of randomness common to all parties. Hence, the standard methods aiming at the detection of quantum nonlocality with linear Bell tests are typically useless when discussing networks, and radically novel methods and concepts must be developed, see e.g. \cite{wolfe2019inflation,MiguelElie,wolfe2021inflation,Aberg2020,Weilenmann_2018,Henson_2014,renou2019limits,gisin2020constraints}. 

Despite these challenges, recent works have brought significant insight into quantum nonlocality in networks. Notably, the few known examples of quantum nonlocal distributions show that the network scenario allows for novel forms of quantum nonlocality, that are possible only due to the network structure (in particular due to the independence  of the sources). Remarkably, quantum nonlocality can be demonstrated in a network where none of the parties receive an input \cite{Fritz_2012,Branciard_2012}, that is all the observers perform a single (fixed) measurement. This is in sharp contrast to the standard Bell scenario, where the presence of measurement inputs is fundamental. More specifically, in a scenario without inputs, the set of correlations obtainable with a classical source common to all parties coincides with the whole set of valid probability distributions. Consequently, there is no nonlocality without inputs in the standard Bell scenario.

A notable example of this phenomenon termed ``quantum nonlocality without inputs'' has been provided by Fritz in the triangle network \cite{Fritz_2012}, as shown in Fig. \ref{fig:triangle}. While the example of Fritz can be viewed as an embedding of the well-known CHSH Bell test in the triangle network (see also \cite{Fraser}), more recent work by Renou et al. \cite{Renou_2019}, followed by other examples \cite{Gisin_2019,Renou2022,Abiuso2022,Pozas2022}, suggested that a different form of quantum nonlocality without inputs can be also observed. Here, no obvious connection to standard Bell nonlocality can be made, suggesting a form of quantum nonlocality genuine to this network structure, as further examined in Ref. \cite{Supic2022}.

\begin{figure}[t]
\centering
\includegraphics[width=0.7\columnwidth]{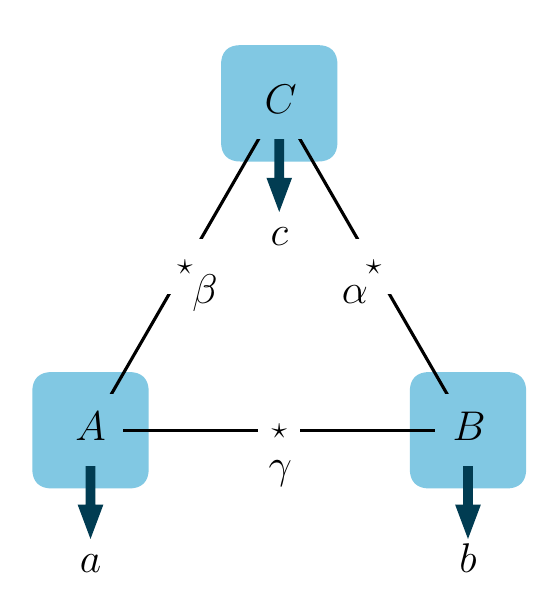}
\caption{The triangle network features three distant parties. Each pair of parties is connected via a bipartite source. Importantly here, all parties perform a fixed measurement, i.e. they receive no input. In this work, we look for a minimal example (in terms of output cardinality) of quantum nonlocality.}
\label{fig:triangle}
\end{figure}

In this work, we investigate the question of finding a minimal example of quantum  nonlocality without inputs. In the study of Bell nonlocality, the notion of minimal usually relates to the number of measurement inputs and outputs of the parties. In particular, the minimal scenario is that of the CHSH Bell inequality \cite{CHSH} where two parties perform two binary measurements each. It has been largely investigated and enables most of the applications of quantum nonlocality, such as device-independent quantum information processing, and most of the experimental demonstrations of quantum nonlocality.

In the context of networks, the notion of minimality can be defined in different ways. First, we aim to minimize the number of parties. Hence we focus on networks with three parties. Among these, the only network that allows for nonlocality without inputs is the triangle network\footnote{Note that the other non-trivial network with three parties, the so-called bilocality network, does not allow for nonlocality without inputs; any distribution compatible with the bilocality network admits a local model~\cite{Branciard_2012}} (see Fig. \ref{fig:triangle}). Next, as they are no measurement inputs, we want to minimize the number of measurement outputs. Previous examples of quantum nonlocality without inputs in the triangle network feature four-valued outputs. We use the notation $4-4-4$ to denote the cardinality of such a scenario. To the best of our knowledge, the only improvement reported so far is a variant of the Fritz example given in Ref. \cite{Weilenmann_2018}, which features output cardinalities of $4-4-2$. 

Here we report on several examples of quantum nonlocality in the triangle network with output cardinalities $3-3-3$ and $3-3-2$. These examples are constructed from coarse-graining of previous known examples. We provide analytic proofs of the nonlocality of these behaviors, and study their robustness to noise via numerical methods based on neural networks \cite{krivachy_neural_2020}. In particular, to prove the nonlocality of the $3-3-3$ example we derive a rigidity result for network local models that we call Parity Token Counting. We believe this is of independent interest, as it represents the minimal example (with cardinalities $2-2-2$) of rigidity of a classical distribution. Finally, we discuss the prospects of finding an example of quantum nonlocality in the simplest triangle network with binary outputs.

\section{Problem and methods}

The triangle network consists of three parties (Fig~ \ref{fig:triangle}), namely Alice, Bob, and Charlie. Each pair of parties is connected by a bipartite source. Importantly, the three sources are assumed to be independent of each other. Formally, independent sources prepare mutually independent random variables in the classical case and product quantum states in the quantum case. Importantly, the three parties do not have access to any common source. Upon receiving the physical resources the observers produce outputs ($a$, $b$ and $c$). 
In contrast to standard tests of Bell nonlocality, the observers receive no input in this setting. The statistics of the outputs are therefore given by the joint probability distribution $P (a, b, c)$. 
Despite the simplicity of the triangle network, finding quantum distributions that are provably nonlocal is a non-trivial task.

One of the central challenges in the study of network nonlocality is to find efficient methods for characterizing the set of distributions $P(a,b,c)$ which are obtainable from different physical resources. In particular, the local set $\mathcal{L}$ consists of all distributions of the following form
\begin{align} \label{trilocal}
P_L(a,b,c)=
\int &\dd\alpha\,   \dd\beta \, \dd \gamma    P_A(a|\beta, \gamma) \, P_B(b|\gamma,\alpha) \, P_C(c|\alpha,\beta),
\end{align}
where $\alpha, \beta$,$\gamma$ are the local variables  distributed by the sources (here $\dd \alpha \equiv \dd \mu(\alpha)$ means that the local variables are sampled from some underlying distributions $\mu(\alpha)$), and $P
_A(a|\beta, \gamma),$ $P_B(b|\gamma,\alpha)$ and $P_C(c|\alpha,\beta)$ are the response functions of Alice, Bob, and Charlie, respectively. To prove that a distribution $P_Q(a,b,c)$ is not ''triangle-local'', we thus need to demonstrate that it cannot be generated in the triangle network with classical sources.
Formally, we need to show that $P_Q(a,b,c)\notin \mathcal{L}
$.

The fundamental difficulty in verifying the existence of a local model according to Eq. (\ref{trilocal}) arises from the independence of the sources. Due to this independence, the local set $\mathcal{L}$ is non-convex\footnote{For a concrete example, consider the distribution where all parties always output ''0'', $P(000)=1$, which is clearly triangle-local. Similarly, the distribution where all parties output ''1'', $P(111)=1$, is triangle-local. However, the mixture of the two distributions violates the no-signaling condition~\cite{wolfe2019inflation}, and is therefore not achievable in the triangle network.}, 
and the standard approach of characterizing the local set (a convex polytope) with linear Bell inequalities (facets of the polytope) cannot be applied. As a result, efficient bounds on the set of classical correlations (for example in the form of nonlinear Bell inequalities) are still missing. Due to the lack of general tools, the known results on quantum network nonlocality without inputs are based on different arguments.

The example of Fritz is constructed by embedding a bipartite CHSH Bell test in the triangle network. It only involves a single quantum source, while the other two sources distribute the "inputs" for the measurement of the entangled quantum state. The nonlocality of the resulting distribution essentially follows from the violation of the CHSH inequality (or any other bipartite Bell inequality)~\cite{Fritz_2012,Fraser}.
 
The example of Renou et al.~\cite{Renou_2019} which we refer to as the RGB4 distribution involves three quantum sources and entangled measurements. The proof of nonlocality of the RGB4 distribution relies on the concept of Token Counting rigidity \cite{Renou2022,renou2022network}. It imposes severe constraints on the underlying classical model for certain triangle-local distributions, as discussed in the next section.
It is also worth mentioning the construction of Gisin in Ref.~\cite{Gisin_2019} which is conjectured to be nonlocal based on numerical evidence~\cite{krivachy_neural_2020}, but for which a proof is still missing.

In this work, we look for minimal examples of quantum nonlocality in the triangle network. More specifically, we present examples of quantum distributions with low output cardinality and show that they do not admit a triangle-local model of Eq.~\eqref{trilocal} using a variety of methods. For our first example in Section III, with output cardinalities $3-3-3$, we provide an analytic proof based on a relaxed version of the rigidity property, which we term ``almost-rigidity''. Our second example in Section IV, with output cardinalities $3-3-2$, is proven via the technique of inflation \cite{wolfe2019inflation}. Loosely speaking, this method considers larger (inflated) networks, some marginals of which coincide with the original triangle network. If one can show that the distribution on the inflated network cannot be consistently defined, given the marginals and some independence constraints, one can conclude that the original distribution is not triangle-local. This technique also allows us to derive nonlinear Bell inequalities for detecting the corresponding distribution. More details can be found in Section IV and in Appendix \ref{appendix_inflation}.

Finally, in order to investigate further examples of quantum nonlocality in the triangle scenario, we use numerical methods based on machine learning. Concretely, we use a method developed in Ref. \cite{krivachy_neural_2020} based on neural networks (NN). The code used in this work is adapted from the code available in this last reference. Here, a generative neural network (which we refer to as "LHV-Net") is used to explore the space of triangle-local models. This is ensured by encoding the structure of the triangle network into the neural network. Given a target distribution $P_{target}(a,b,c) $, the LHV-Net aims at constructing a local distribution $P_{NN}(a,b,c)$ which is as close as possible to the target one. This is done by minimizing the loss function, which in most cases is the Euclidean distance  \begin{equation}
\begin{aligned}
    \label{distance_formula}
     d (P_{target},& P_{NN}) = \\ &\sqrt{\sum_{a,b,c} \left[P_{target}(a,b,c)-P_{NN}(a,b,c) \right]^2}.
\end{aligned}
\end{equation}

For a given target distribution, it is often useful to investigate a family of distributions obtained by adding noise to the initial one. This can be done by mixing the target distribution with another distribution that is triangle-local (for example by adding noise to the quantum model). The visibility parameter $V$ controls the amount of noise, with $V=1$ corresponding the initial distribution. If the initial distribution $P_{target}$ is indeed nonlocal, at some value $V^*$ we expect to see a sharp transition when monitoring the distance $d(\cdot, \cdot)$ as a function of $V$. The critical visibility $V^*$ then corresponds to the point where the noisy distribution becomes triangle-local and gives an estimate of the noise robustness of the initial distribution.

\section{Quantum Nonlocality with output cardinality 3-3-3}

We start by presenting an example of quantum nonlocality in the triangle network with output cardinalities $3-3-3$. This example is based on a coarse-graining of the RGB4 family of distributions (with biased tokens/sources) of Ref. \cite{Renou_2019}, which has cardinalities $4-4-4$. In order to demonstrate the nonlocality of the coarse-grained distribution, we first present a result on ``almost-rigidity'' for the task of Parity Token Counting. Then we give the target quantum distribution and use the almost-rigidity property to prove that it is not triangle-local.

\subsection{Almost-rigidity of Parity Token Counting distributions}

In the present context, rigidity~\cite{Renou2022,renou2022network} is a property of some classes of classical distributions in a network. The property ensures that any model (i.e. local variables and response functions displayed on the r.h.s of Eq.~\eqref{trilocal}) underlying a rigid probability distribution $P_L(a,b,c)$ can be brought to a unique canonical form using local relabeling of the variables $\alpha, \beta, \gamma$ (to a discrete set of values) which do not conflict with the response function. To be more precise, we now  elaborate on the general form of triangle-local models in Eq.~\eqref{trilocal} and Fig.~\ref{fig:triangle}. 

To fix the notation we say that each source samples local variables $\alpha, \beta$ or $\gamma$ according to distributions $\mu(\alpha), \eta(\beta)$ and $\nu(\gamma)$. Note that without loss of generality, we assumed that each source distributes two copies of the same variable\footnote{One can also consider sources sampling two variables that may take different values and sending them to different parties, e.g. $(\alpha_B,\alpha_C)$ with $\alpha_B$ sent to Bob and $\alpha_C$ to Charlie. This case is however included in the previous one, as the source could alternatively distribute two copies of pair $\alpha=(\alpha_B,\alpha_C)$ to both parties while the response functions ignore the respective half of the pair.}. Moreover, without loss of generality, the response functions can be considered deterministic, i.e. $P_A(a|\beta,\gamma)=0$ or $1$, because any randomness required for the choice of the outputs (local to the party) can be delegated to one of the neighboring sources. Hence, without loss of generality, we use the three distributions of local variables $\mu(\alpha),\eta(\beta),\nu(\gamma)$ and the three deterministic response functions $a(\beta,\gamma),b(\alpha,\gamma), c(\alpha,\beta)$ to describe a generic model underlying triangle-local distributions $P_L(a,b,c)$.

Let us now construct a family of triangle-local models that we call Parity Token Counting (PTC) strategies. Each source has a single token that can be sent to either one or the other party it connects to. For example, the source connecting Bob and Charlie prepares $(\alpha_B,\alpha_C)= (1,0)$ with probability $p_\alpha$ and $(\alpha_B,\alpha_C)=(0,1)$ with probability $1-p_\alpha$. Each party then outputs the parity of the total number of tokens it receives, i.e. for Charlie $c(\alpha_C,\beta_C)= \alpha_C \oplus \beta_C$, where ``$\oplus$'' is understood as sum modulo two, so that all the outputs are binary. All possible PTC strategies we just described give rise to distributions $P_{PTC}(a,b,c)$ that live in a three-dimensional subset $\mathcal{L}_{PTC}$ of the local set for the $2-2-2$ triangle. The set $\mathcal{L}_{PTC}$ can be parameterized by $(p_\alpha, p_\beta,p_\gamma)$, however, the strategies $(p_\alpha, p_\beta,p_\gamma)$ and $(1-p_\alpha, 1- p_\beta, 1- p_\gamma)$ are related by flipping the values of all local variable and thus lead to the same distribution. Except for a subset of $\mathcal{L}_{PTC}$ of measure zero these two strategies are the only ones that simulate the distribution. The full characterization of the set $\mathcal{L}_{PTC}$, including the multiplicities of the PTC strategies underlying each distribution $P_{PTC}(a,b,c)$, can be found in Appendix.\ref{Appendix_PRCdistributions}.

We now ask what possible models lead to a distribution $P_{PTC}(a,b,c)$, except for the strategies we just defined. In fact, we will now show that PTC strategies are essentially the only classical  strategies that lead to a distribution $P_{PTC}(a,b,c)$. This is formalized by the following theorem.
\begin{theorem}
 Let $P_{PTC}(a, b, c)$ be the distribution arising from Parity Token Parity Counting strategy $(p_\alpha, p_\beta, p_\gamma)$ on the triangle network, where each source distributes its only token to the connected parties with probabilities $p_i$ and $1-p_i$, and the parties output the parity of the number of received tokens. For any other strategy $\big(\mu(\alpha),\eta(\beta),\nu(\gamma),a(\beta,\gamma),b(\alpha,\gamma), c(\alpha,\beta) \big)$ that achieves the distribution $P_{PTC}(a, b, c)$ there exist functions $T_i^j: \mathcal S_i\to \{0,1\}$ ($\mathcal S_i$ is the set of all possible values of the local variable produced by the source $S_i$) for any $S_i\to A_j$ such that 
\begin{enumerate}
\item[{\rm (i)}] 
\resizebox{0.9\hsize}{!}{
$T_{\alpha}^b(\alpha)+T_{\alpha}^c(\alpha) = T_{\beta}^c(\beta)+T_{\beta}^a(\beta) = T_{\gamma}^a(\gamma)+T_{\gamma}^b(\gamma) = 1.$}
\item[{\rm (ii)}] $a (\beta, \gamma) =T_{\beta}^a(\beta) \oplus T_{\gamma}^a(\gamma),\\
b (\gamma , \alpha) = T_{\gamma}^b(\gamma) \oplus T_{\alpha}^b(\alpha),\\
c (\alpha , \beta) =T_{\alpha}^c(\alpha) \oplus T_{\beta}^c(\beta).$
\end{enumerate}
In addition, if the distribution is such that $P(a=1)$, $P(b=1)$, $P(c=1)\neq \frac{1}{2}$, the functions $T_{i}^j$ can be chosen to fulfill
\begin{itemize}
    \item[{\rm (iii)}]
$\left(
\mathds{E}[T_{\alpha}^b],
\mathds{E}[T_{\beta}^c],
\mathds{E}[T_{\gamma}^a]
\right) =
\left(
p_\alpha, 
p_\beta, 
p_\gamma
\right)$
\end{itemize}
where $\mathds{E}[T_{\alpha}^b] = \int  \dd\alpha \,  T_{\alpha}^b(\alpha)$.
\end{theorem}

The proof is given in Appendix \ref{Appendix_SemiRig}. We now discuss how to interpret it.
First, note that the six local compression functions $T_{i}^j$ map the local variable at each output of each source into a bit, e.g 
\begin{equation}
    \begin{split}
        (T_\alpha^b,T_\alpha^c):\mathcal{S}_\alpha &\to \{0,1\}^{\times 2}\\
        \alpha &\mapsto (T_\alpha^b(\alpha),T_\alpha^c(\alpha)),
    \end{split}
\end{equation}
for the two outputs of the source $S_\alpha$.
Furthermore, the condition (i) guarantees that for any $\alpha$ the only possible values are  $(\alpha_B,\alpha_C)=(T_\alpha^b(\alpha),T_\alpha^c(\alpha))= (1,0)$ or $(0,1)$, so that after the application of the compression functions to the outputs of a source it becomes a PTC source. On top of this, condition (ii) guarantees that the response functions of the original strategy are consistent with such compression, e.g. $c (\alpha , \beta) = \alpha_C\oplus \beta_C = T_{\alpha}^c(\alpha) \oplus T_{\beta}^c(\beta)$. That is, the response functions can  only depend on the local variables via the values of the bits after compression. In other words, any information encoded in $\alpha$ other than the single bit given by $(T_\alpha^b(\alpha),T_\alpha^c(\alpha))$ is ignored by the response functions.

Finally (iii) guarantees that in most cases there is a unique  PTC strategy (up to the choice of the functions $T_i^j$) underlying each PTC distribution. Nevertheless, for the following discussion, we will be interested in distributions with unbiased outputs for which (iii) does not apply. The cases where at least one output is unbiased (a subset of $\mathcal{L}_{PTC}$ of measure zero) are degenerate, i.e. there are infinitely many PTC strategies that lead to say same distribution, see Appendix~\ref{Appendix_PRCdistributions}. Since (iii) does not hold for all distributions we refer to Theorem 1 as almost-rigidity.

Two things are worth mentioning before moving on. First, the almost-rigidity property of PTC distributions holds for all networks satisfying the No Double Common Source (NDCS) condition~\cite{Sadrainprep}. Second, PTC rigidity works for binary outputs, which is in contrast to Token Counting and Color Matching distributions requiring at least ternary outputs~\cite{renou2022network}. So in some sense, it provides the minimal example of network rigidity. This is also why it can be used to detect nonlocality in the $3-3-3$ triangle as we now show.

\subsection{The quantum strategy}

Consider each source distributing an entangled two-qubit state 

\begin{equation}
\ket{\psi} = \lambda_0 \ket{01} + \lambda_1\ket{10}, \label{state_RGB4}
\end{equation}

with positive real coefficients $\lambda_0, \lambda_1$ where $\lambda_0^2+\lambda_1^2=1$.
Each party performs the following two-qubit joint measurement with ternary outputs $a,b,c\,{\in}\,\{\bar{0}, \bar{1}_0, \bar{1}_1 \}$.
\begin{equation} \label{meas_RGB4}
  \begin{split}
       \bar{0} &: \ketbra{00}{00} + \ketbra{11}{11} \\
       \bar{1}_0 &: \ketbra{\bar{1}_0}{\bar{1}_0} \\
       \bar{1}_1 &: \ketbra{\bar{1}_1}{\bar{1}_1} \\
  \end{split}
\end{equation}
where we define, $\ket{\bar{1}_0} =  u\ket{01}+v\ket{10},~ \ket{\bar{1}_1} = v\ket{01}-u\ket{10}$, with $0<v<u<1$ real parameters that satisfy $u^2\,{+}\,v^2\,{=}\,1$. This strategy is similar to the RGB4 construction \cite{Renou_2019}, with the difference that the POVM elements $\ketbra{00}{00}$ and $\ketbra{11}{11}$ are here coarse-grained into a unique output $\bar{0}$. 

Fixing $\lambda_0^2=\frac{1}{3}$ and combining the above states and measurements, we obtain a quantum distribution denoted $P_Q^{333}(a,b,c)$ and given in Appendix.~\ref{Appendix_RGB4_proof}, for which we can prove the following.

\begin{theorem}\label{theorem_parity_qubit}
\label{theorem2}
The quantum distribution $P_Q^{333}(a,b,c)$ is not triangle-local (incompatible with any model of the form ~\eqref{trilocal}) for 
$\frac{2}{3}< u^2<1$.
\end{theorem}

Below we sketch the proof of this Theorem, while all details are given in Appendix \ref{Appendix_RGB4_proof}, where we identify a range of parameters $\lambda_0^2$ and $u$ that lead to triangle-nonlocal distributions (See Fig.~\ref{function} in the appendix).  The structure of the proof is similar to Ref. \cite{Renou_2019}, but it is based on the novel PTC rigidity result which is essential to make the proof work for ternary outputs. First, we observe that if the outputs $\bar{1}_0$ and $\bar{1}_1$ are coarse-grained for all parties the distribution resulting from $P^{333}_Q(a,b,c)$ is triangle-local and PTC. The rigidity of PTC distributions (Theorem 1) imposes severe constraints on the underlying classical models. Now, it can be shown that if the distribution $P^{333}_Q(a,b,c)$ was simulated by a classical model (before coarse-graining) these constraints would be impossible to satisfy. Therefore $P^{333}_Q(a,b,c)$ is necessarily not triangle-local.

\subsection{Noise robustness}

The above result applies to the (noiseless) distribution $P_Q^{333}(a,b,c)$, but does not extend to the case where noise is added, as the coarse-grained version of the noisy distribution is no longer PTC in general. Thus, to investigate noise robustness we must resort to other methods.

A first possibility would be to use the inflation method. Unfortunately, using it we could not prove the nonlocality of the above $3-3-3$ quantum distribution. Note that the inflation method is able to detect nonlocality of the original RGB4 distribution, albeit with an extremely small noise tolerance~\footnote{Elie Wolfe, private communication}.

Instead, we move to numerical methods and use LHV-Net. For the noise model, we add white noise to the entangled states of Eq.~\eqref{state_RGB4} produced by each source, which now prepares Werner states
\begin{align} \label{werner}
\rho_V = V \ket{\psi^+}\bra{\psi^+} + (1-V) \openone/4.
\end{align}
Leaving the measurement in Eq.~\eqref{meas_RGB4} unchanged, we obtain a family of distributions $P_{Q|V}^{333}$.  

In Fig. \ref{fig:RGB4_plots}, we plot the minimal distance (see Eq. (\ref{distance_formula})) found by LHV-Net, when taking $P_{Q|V}^{333}$ for the target distribution. We observe that, as $V$ decreases, the minimal distance undergoes a transition around the critical visibility $V^* \approx 0.99$, demonstrating a small robustness to noise. Note that, for the original RGB4 distribution (with four-valued outputs), the estimated noise tolerance (also via LHV-Net) is much larger \cite{krivachy_neural_2020}. Hence it seems that the coarse-graining severely decreases the robustness to noise in this case.

\begin{figure}[t]
\centering
\includegraphics[width=\columnwidth]{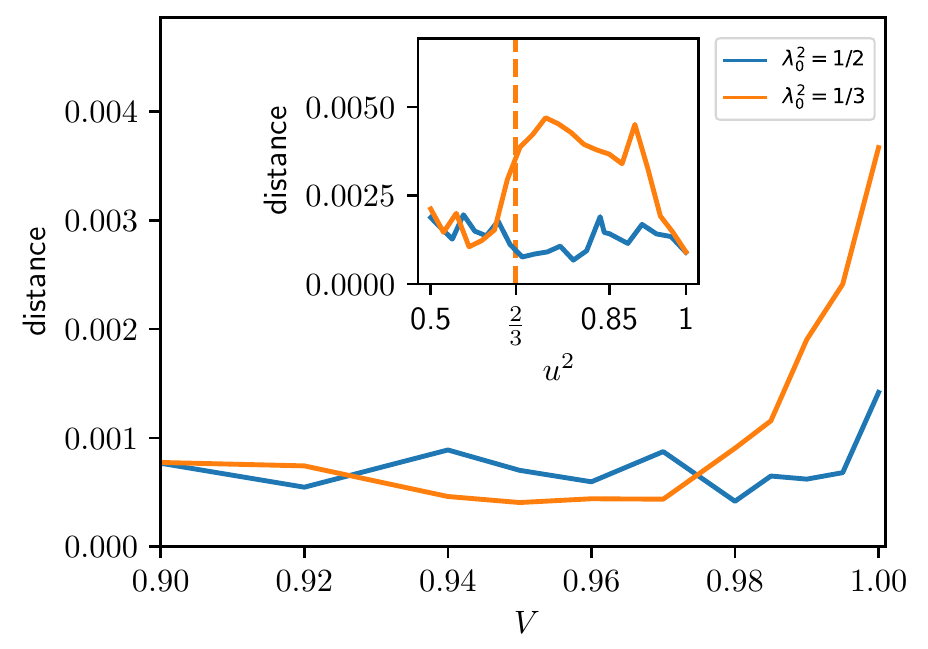}
\caption{Noise robustness of the $3-3-3$ coarse-grained RGB4 quantum distribution. The plot shows the minimal distance $d(P_{target}, P_{NN})$ found by the neural network (LHV-Net) between the target distribution and the closest local distribution, for different levels of noise, quantified by the visibility parameter $V$. Here, the measurements are given by setting $u^2 = 0.85$, and we consider two values for $\lambda_0$, the orange curve is for $\lambda_0=1/3$, and the blue curve for $\lambda_0=1/2$ for which our proof technique stays silent about the nonlocality of the distribution.  The inset demonstrates the distance for different values of the measurement parameter $u^2$, and the vertical line at $2/3$ represents the critical point for the orange curve for which we have proven nonlocality for $2/3<u^2<1$. We can see from this plot that the resulting distribution with $\lambda_0=1/2$ is  very close to the local set for all values of $u^2$. However, the distribution for $\lambda_0=1/3$ is much further from the local set and seems to be noise robust up to visibility $V=0.97$}
\label{fig:RGB4_plots}
\end{figure}

Moreover, LHV-Net also allows us to investigate different coarse-grainings of the RGB4 distribution. In particular, we find that when combining the outputs in the following way: $0: \ketbra{00}{00} + \ketbra{\bar{1}_0}{\bar{1}_0}$, $1: \ketbra{11}{11}$ and $2: \ketbra{\bar{1}_1}{\bar{1}_1} $, the resulting $3-3-3$ distribution also appears to be nonlocal. The noise robustness is similar to the previous case. It is also worth mentioning that the coarse-graining $0: \ketbra{00}{00}$, $1: \ketbra{\bar{1}_0}{\bar{1}_0} + \ketbra{\bar{1}_1}{\bar{1}_1} $ and $2: \ketbra{11}{11}$ results in product measurements and gives a classical distribution.

Finally, we also investigated binary coarse-grainings of the RGB4 distribution, but according to LHV-Net, all resulting distributions appear to be triangle local. Of course, this is only a numerical result and it should be considered carefully. It could be in principle that the resulting distribution is nonlocal but extremely close to the set of triangle-local distributions.

\subsection{Elegant distribution}

Using LHV-Net, we also investigated coarse-graining of the so-called ``Elegant distribution'' proposed by Gisin~\cite{Gisin_2019}. The distribution is based on all sources producing a Bell states, here chosen to be the singlet state $\ket{\psi^-}=1/\sqrt{2}(\ket{01}-\ket{10})$. Each party performs an entangling measurement with the four outcomes corresponding to projectors $\{\ketbra{\Phi_1}{\Phi_1}, \dots, \ketbra{\Phi_4}{\Phi_4}\}$ on two-qubit states 
\begin{equation}
\ket{\Phi_j}=\sqrt{\frac{3}{2}}\ket{m_j, -m_j} + i \frac{\sqrt{3}-1}{2}\ket{\psi^-}, 
\label{elegant_measurement}
\end{equation}
with the four vectors $\{\ket{m_j}\}_{j=1,2,3,4}$ forming a tetrahedron on the Bloch sphere. A notable difference with the RGB4 measurement in Eq.~\eqref{meas_RGB4}, is that here the states $\ket{\Phi_j}$ corresponding to all measurement outcomes are entangled, with the same level of entanglement (they are all related by local unitary transformations). The resulting distribution is symmetric under permutation of outputs and parties and is given by 
\begin{equation}
\begin{split}
P(a=b=c)&=\frac{25}{256}\\
P(a=b\neq c)&=\frac{1}{256}\\
P(a\neq b\neq c \neq a) &= \frac{5}{256} \,.
\end{split}
\label{def_elegant}
\end{equation} 
There is currently no known proof that the elegant distribution is not triangle-local. 

 Given the level of symmetry of this distribution, it is enough to consider a single coarse-graining to a $3-3-3$ distribution, with the constraint that each party performs the same coarse-graining. Here LHV-Net also suggests that the resulting distribution remains nonlocal, as shown in Fig. \ref{fig:Elegant_plots}. Again, it seems that the noise tolerance is reduced compared to the original $4-4-4$ distribution. 

Finally, note that when using LHV-Net to analyze different ternary coarse-grainings for the parties or binary coarse-grainings, it appears that all resulting distributions are triangle-local.

\begin{figure}[t]
\centering
\includegraphics[width=\columnwidth]{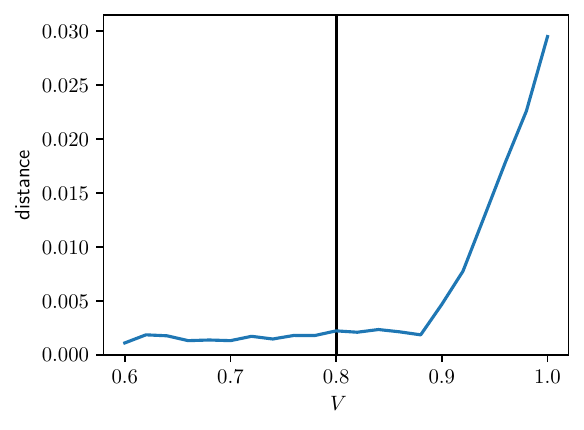}
\caption{Noise robustness for the coarse-grained Elegant distribution. The plot shows the distance $d(P_{target}, P_{NN})$ between the target distribution and the closest local distribution found by LHV-Net, for different values of the visibility $V$. The vertical line represents the critical visibility for the original four-output distributions as found in \cite{krivachy_neural_2020}. Again, it seems that coarse-graining reduces robustness to noise.  }
\label{fig:Elegant_plots}
\end{figure}

\section{Quantum nonlocality with output cardinality 3-3-2}

We now present an example of quantum nonlocality in the triangle that uses even smaller cardinalities for the outputs, namely $3-3-2$. This example is constructed from coarse-graining the Fritz distribution \cite{Fritz_2012}.

Let us first recall the model of Fritz. The idea is to have Alice and Bob perform a standard CHSH Bell test~\cite{CHSH}. Hence they share a singlet Bell state $\ket{\psi^-}$. Now, the binary measurement inputs that are required for both Alice and Bob for testing the CHSH inequality are provided by the two extra sources. Specifically, the $\beta$ source provides a uniformly random bit $x=0,1$ while the $\alpha$ source produces a uniformly random bit $y=0,1$. Upon receiving their effective inputs $x$ and $y$, Alice and Bob perform the corresponding local Pauli measurements ($\sigma_z$ or $\sigma_x$ for Alice and $(-\sigma_z-\sigma_x)/\sqrt{2}$ or $(-\sigma_z+\sigma_x)/\sqrt{2}$ for Bob) and obtain binary outputs $a'$ and $b'$. Finally, Alice outputs $a=(a',x)$, Bob $b=(b',y)$ and Charlie $c=(x,y)$. Note that it is crucial that Charlie also broadcasts both effective inputs $x$ and $y$ in order to ensure the condition of measurement independence. 

Fritz showed that since the values of the outputs $x(y)$ are perfectly correlated between Charlie and Alice (Charlie and Bob), these outputs have to be independent of the source connecting Alice and Bob~\cite{Fritz_2012}. Therefore, the resulting distribution $P_F(a,b,c)$ is nonlocal whenever the conditional distribution $P(a',b'|x,y)$ violates the CHSH Bell inequality. When adding noise to the singlet state shared by Alice and Bob (similarly to Eq. \eqref{werner}), one finds a critical visibility of $V^* = 1/\sqrt{2}$. 

We first consider the following coarse-graining. On Alice's side, the output is ternary and given by: $a=0$ if $a'=0$, $a=1$ if $x=0$ and $a'=1$, and $a=2$ if $x=1$ and $a'=1$. Similarly, on Bob's side, we define a ternary output according to: $b=0$ if $x=0$ and $b'=0$, $b=1$ if $y=1$ and $b'=0$, and $b=2$ if $b'=1$. Finally, on Charlie, we get a binary output: $c=xy$, i.e. $c=1$ only if $x=y=1$. Note that we consider here uniform distributions for both effective inputs to be biased, specifically $p(x=1) = p(y=1) = 1/2$, which renders the binary output $c$ biased. We have checked that changing these input distributions only marginally influence the result.

The resulting distribution, with cardinality $3-3-2$ is nonlocal, which we can prove using the inflation technique. Specifically, we consider the so-called web inflation, which allows us to construct a nonlinear inequality for detecting the above coarse-grained distribution for visibility greater than $ \approx 0.87$. This value is higher compared to the critical visibility of the original Fritz distribution. This appears to be a limitation of our inflation, since using LHV-Net, we find that the coarse-grained $3-3-2$ distribution still has a critical visibility close to the original value $1/\sqrt{2}\simeq 0.71$. These results are shown in Fig. \ref{fig:Fritz_plots}.

\begin{figure}[t]
\centering
\includegraphics[width=\columnwidth]{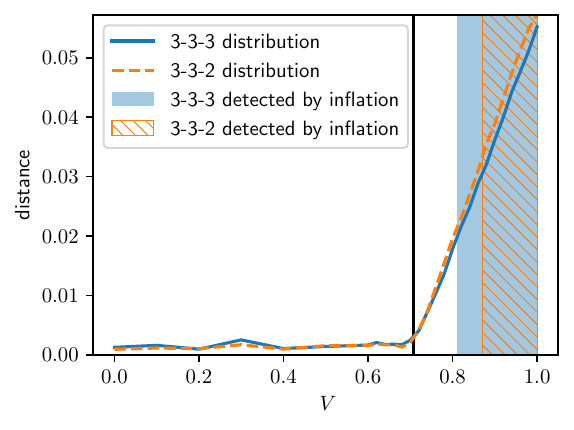}
\caption{Noise robustness of the coarse-grained Fritz distributions. The plot shows the distance $d(P_{target}, P_{NN})$ between the target distribution and the closest local distribution found by LHV-Net, for different values of the visibility $V$. Both $3-3-2$ and $3-3-3$ cases are shown. The vertical line represents the critical visibility ($V^*=1/\sqrt{2}$) for the original four-output distributions, see \cite{krivachy_neural_2020}.
The colored regions indicate the visibilities detected via the web inflation, see Appendix \ref{appendix_inflation} for details. }
\label{fig:Fritz_plots}
\end{figure}

More generally, there exist many different possible coarse-grainings of the Fritz distribution that lead to interesting $3-3-2$ or $3-3-3$ distributions. In Appendix \ref{appendix_inflation}, we provide a detailed analysis of all combinations resulting in triangle-nonlocal distributions and discuss their noise tolerance.

Finally, we also tried here to obtain quantum nonlocality from coarse-graining to a binary distribution, but in all cases, LHV-Net can reproduce the resulting distribution with excellent accuracy.

\section{Triangle with binary outputs}

Finding an example of quantum nonlocality in the triangle with binary outputs would of course provide a minimal example of network nonlocality. At this point, it is still an open question whether this is possible or not. We obtained relatively convincing numerical evidence that all examples discussed above become triangle-local when coarse-grained to binary outputs. It also seems that the rigidity-based proofs of nonlocality, along the lines of Theorem \ref{theorem2}, cannot work for binary quantum distributions, since they require to elaborate on a coarse-groaned version of the original distribution.

We also investigated another approach based on a connection to the Lovasz local lemma. The latter states that when considering events that are almost independent and individually not very likely, there is always a non-vanishing probability that none of them occurs. A refined version called variational Lovasz local lemma (VLLL) \cite{he2017variable} applies to events described by a fixed event-variable-graph, or to network-local models with binary outputs in the language of this paper. When applied to the case of the triangle network, Theorem 4 in~\cite{he2017variable} implies that for any triangle-local strategy
\begin{equation}
    P(j=0) < p^*  \quad \forall \, j=a,b,c  \,\, \implies P(a=b=c=1)>0 \label{LLL}
\end{equation}
where $p^* = (\sqrt{5}-1)/2 \approx 0.38$. Interestingly, the above relation is tight and can be saturated with a classical strategy using only two sources\footnote{Consider a source distributing the shared variable $\beta=0,1$ to A and C with probability $P(\beta=1)= p^*$. Similarly, B and C are connected via a source sending $\alpha=0,1$ with probability $P(\alpha=1)= p^*$. Then consider that A outputs $\beta$, B outputs $\alpha$, while C outputs $c=\alpha\beta \oplus 1$. It is straightforward to check that the resulting distribution saturates the condition given in Eq. \eqref{LLL}. That is, it leads to $P(a=0)=P(b=0)=P(c=0)=p^*$ and $P(a=b=c=1)=0$.} (i.e. corresponding to the bilocality network). The implication~\eqref{LLL} is not necessarily true for quantum strategies. Moreover, finding a quantum example $P_Q(a,b,c)$ that violates the VLLL (the condition~ \eqref{LLL}) would  reveal nonlocality in the $2-2-2$ triangle. We have investigated this question numerically, but could not find an instance of a quantum violation.

\section{Discussion}

We investigated the question of finding a minimal example of quantum nonlocality without inputs in a network. We focused on the triangle network, as this involves the minimal number of parties, i.e. three. We found two classes of examples involving respectively output cardinalities $3-3-3$ and $3-3-2$. The examples were constructed from coarse-graining of previously known examples. To prove some of our results, we considered the scenario of Parity Token Counting, for which we proved an almost-rigidity property, which can be of independent interest as it represents a minimal example of rigidity of a classical distribution. We also discussed the noise robustness of these examples.

The main question left open here is whether there exists an example of quantum nonlocality in the minimal triangle network with binary outputs. For all examples we found, it appears that further coarse-graining to obtain binary outputs leads to triangle-local distributions. We also established a connection to the Lovasz local lemma, but could  not find a quantum nonlocal distribution with this approach. We note, however, that a slightly more complicated network, namely a ring network with four parties (i.e. a square) allows for quantum nonlocality without inputs and binary outputs \cite{Restivo}.

In addition to the binary outputs question, it would be interesting to find an  example of nonlocality without inputs in the $3-2-2$ triangle network. Here the set of correlation $P(a,b,c)$ is described by $11=3\times2\times 2-1$ real parameters. This is lower than the $12= 2\times2\times2\times 2-4$ parameters needed to describe the set of correlations $P(a,b|x,y)$ in the CHSH scenario. Hence, from the perspective of the dimension of the correlation set, an example of quantum nonlocality in the triangle with output cardinality $3-2-2$ would be ''more minimal'' than the minimal example of standard Bell nonlocality.

Finally, another interesting question is whether the $3-3-3$ examples we constructed here feature some stronger form of network correlations, for example genuine network nonlocality \cite{Supic2022} or full network nonlocality \cite{Pozas_full}. Since our $3-3-2$ example is constructed from the Fritz distribution, these can be neither genuine network nonlocal nor full network nonlocal.

\medskip

\emph{Acknowledgements.---} We thank Marc-Olivier Renou for the discussions. We acknowledge financial support from the Swiss National Science Foundation (project 2000021\_192244/1 and NCCR SwissMAP).

\onecolumngrid

\appendix

\setcounter{theorem}{0}

\section{The set $\mathcal{L}_{PTC}$  of Parity Token Counting distributions in the $2-2-2$ triangle} \label{Appendix_PRCdistributions}

PTC local models defined in the main text are constructed with bipartite sources sampling
\begin{equation}\begin{split}
(\alpha_B, \alpha_C) &= \begin{cases} (1,0) & \text{with probability } \quad p_\alpha \\
(0,1) & \text{with probability} \quad 1- p_\alpha
\end{cases}\\
(\beta_C, \beta_A) &= \begin{cases} (1,0) & \text{with probability } \quad p_\beta \\
(0,1) & \text{with probability} \quad 1- p_\beta
\end{cases}\\
(\gamma_A, \gamma_B) &= \begin{cases} (1,0) & \text{with probability } \quad p_\gamma \\
(0,1) & \text{with probability} \quad 1- p_\gamma
\end{cases}.
\end{split}
\end{equation}
and fixed binary response functions
\begin{equation}
    a(\beta_A,\gamma_A) = \beta_A\oplus\gamma_A \qquad   b(\gamma_B,\alpha_B) = \gamma_B \oplus \alpha_B \qquad c(\alpha_C,\beta_C) = \alpha_C \oplus \beta_C.
\end{equation}

Obviously, the set $\mathcal{L}_{PTC}$ of PTC distributions is parameterized by the three values $(p_\alpha, p_\beta, p_\gamma)$ and is at most of dimension 3. $\mathcal{L}_{PTC}$ is a subset of the 7-dimensional local set $\mathcal{L}_{2-2-2}$ for the $2-2-2$ triangle. The exact dependence of the probability distribution $P_{PTC}(a,b,c)$ on the parameters of the sources are given in Table.~\ref{tab:PTC}. From there one can easily verify that $\mathcal{L}_{PTC}$ is indeed 3-dimensional, e.g. by noting that $P_{PTC}(1,0,0),P_{PTC}(0,1,0)$ and $P_{PTC}(0,0,1)$ are linearly independent polynomials of $(p_\alpha, p_\beta, p_\gamma)$. One can also check that the transformation $(p_\alpha, p_\beta, p_\gamma) \to (1-p_\alpha, 1-p_\beta, 1-p_\gamma)$ does not affect the distribution. 

\begin{table}[h]
    \centering
    \begin{tabular}{|c|c c|}
    \hline
         $(b,c)$\textbackslash $a$&  0 & 1   \\
         \hline
         (0,0)& 0 & $p_\alpha (1-p_\beta)(1-p_\gamma)+ (1-p_\alpha) p_\beta p_\gamma$\\
         (0,1)& $p_\alpha p_\beta (1-p_\gamma)+ (1-p_\alpha)(1-p_\beta)p_\gamma$ & 0\\
         (1,0)& $p_\alpha (1-p_\beta) p_\gamma + (1-p_\alpha) p_\beta (1- p_\gamma)$& 0\\
         (1,1)& 0 & $p_\alpha p_\beta p_\gamma+ (1-p_\alpha)(1-p_\beta)(1-p_\gamma)$\\
         \hline
    \end{tabular}
    \caption{The table of probabilities $P_{PTC}(a,b,c)$ as a function of $(p_\alpha, p_{\beta},p_\gamma)$.}
    \label{tab:PTC}
\end{table}

To describe the set $\mathcal{L}_{PTC}$ it is convenient to introduce the one party correlators
\begin{equation}
E_A =2 P(a=1) - 1 \qquad E_B =2 P(b=1) - 1 \qquad E_C =2 P(c=1) - 1,
\end{equation}
which satisfy the equations 
\begin{equation}\label{eq: corr params}
P(a=1) = \frac{1+E_A}{2} = p_\beta p_\gamma +(1-p_\beta)(1-p_\gamma)
\end{equation}
upon party permutations. Given these equations it is straightforward to rewrite the full distribution $P_{PTC}(a,b,c)$ in terms of the correlators $(E_A,E_B,E_C)$.  It is as given in Table.~\ref{tab:PTC corr}. Since any PTC distribution can be reconstructed completely given the values of the correlators, it is natural to ask which values $(E_A,E_B,E_C)$ are actually achievable by some PTC strategy. To answer this question  we will now distinguish three disjoint subsets of $\mathcal{L}_{PTC}$ depending on the values of correlators.  \\

\begin{table}[h]
    \centering
    \begin{tabular}{|c|c c|}
    \hline
         $(b,c)$\textbackslash $a$&  0 & 1   \\
         \hline
         (0,0)& 0 & $\frac{1}{4}(1+E_A-E_B-E_C)$\\
         (0,1)& $\frac{1}{4}(1-E_A-E_B+E_C)$ & 0\\
         (1,0)& $\frac{1}{4}(1-E_A+E_B-E_C)$& 0\\
         (1,1)& 0 & $\frac{1}{4}(1+E_A+E_B+E_C)$\\
         \hline
    \end{tabular}
    \caption{The table of probabilities $P_{PTC}(a,b,c)$ as a function of $(E_A,E_B,E_C)$.}
    \label{tab:PTC corr}
\end{table}

\textbf{1. The generic case $E_A,E_B,E_C \neq 0$.}   Without loss of generality we permute the parties such that $E_A\geq E_B \geq E_C$. Formally solving the equations \eqref{eq: corr params} for all parties gives two solutions
\begin{equation}\label{eq: corr sol}
\left(\begin{array}{c}
p_\alpha \\
p_\beta \\
p_\gamma
\end{array} 
\right)=
\frac{1}{2}\left(\begin{array}{c}
 1 \pm \sqrt{\frac{E_B E_C}{E_A}} \\
 1 \pm \sqrt{\frac{E_A E_C}{E_B}}\\
1 \pm \sqrt{\frac{E_A E_B}{E_C}}
\end{array} 
\right).
\end{equation}
Note that the two solutions are related by the transformation $(p_\alpha, p_\beta, p_\gamma) \to (1-p_\alpha, 1-p_\beta, 1-p_\gamma)$ mentioned above. This degeneracy is thus easy to understand, and we can focus on the solution with the ''+'' sign. 
The solution in Eq.~\eqref{eq: corr sol} makes physical sense if and only if the numbers on the rhs are probabilities, that is $(\ast) \, \sqrt{\frac{E_B E_C}{E_A}} \in \mathds{R}$ and $ (\ast\ast)\, \sqrt{\frac{E_B E_C}{E_A}} \leq 1$ for all party permutations. The reality conditions imply that only an even number of correlators $E_A, E_B, E_C$ can be negative, leading to two different cases $(E_A\geq E_B\geq E_C \geq 0)$ or $(E_A\geq 0 > E_B \geq E_C)$ that fulfill $(\ast)$.
\begin{itemize}
    \item [\textbf{1a.}] For $(E_A\geq E_B\geq E_C > 0)$, the fractions are ordered as $\sqrt{\frac{E_A E_B}{E_C}}\geq \sqrt{\frac{E_A E_C}{E_B}} \geq \sqrt{\frac{E_B E_C}{E_A}}$, hence the condition $(\ast \ast)$ is verified iff $\sqrt{\frac{E_A E_B}{E_C}}\leq 1$ or simply $E_A E_B \leq E_C$.

\item[\textbf{1b.}] For $(E_A> 0 > E_B \geq E_C)$, we $E_C \leq E_B$ and thus $|E_C|\geq |E_B|$. The fractions are partially ordered $\sqrt{\frac{E_A E_B}{E_C}}\leq \sqrt{\frac{E_A E_C}{E_B}} $ and the condition $(\ast \ast)$ is fullfilled iff both $\sqrt{\frac{E_A E_C}{E_B}}, \sqrt{\frac{E_B E_C}{E_A}} \leq 1$. The two inequalities can be rewrittend as $E_A |E_C| \leq |E_B |\leq \frac{E_A}{|E_C|}$.
\end{itemize}
 
 \textbf{ 2. The measure zero case $E_A E_B E_C = 0$.} Without loss of generality let us assume that $E_A=0$. By Eq.~\eqref{eq: corr params} we have
 \begin{equation}
     \frac{1}{2} = p_\beta p_\gamma +(1-p_\beta)(1-p_\gamma) \implies p_\beta = \frac{1}{2} \quad \text{or} \quad p_\gamma =\frac{1}{2}.
 \end{equation}
 Let us consider the case where $p_\beta=\frac{1}{2}$. This automatically implies that $E_C=0$ as well. Furthermore, with the help of the Table.~\ref{tab:PTC} one computes the full probability distribution  
 \begin{equation}\begin{split}
     P_{PTC}(1,0,0)&=P_{PTC}(0,0,1) =\frac{1}{2}(p_\alpha (1-p_\gamma)+ (1-p_\alpha) p_\gamma )=\frac{1-E_B}{4} \\
     P_{PTC}(0,1,0)&=P_{PTC}(1,1,1) = \frac{1}{2}(p_\alpha p_\gamma + (1-p_\alpha)  (1- p_\gamma)) = \frac{1+E_B}{4}.
 \end{split}
 \end{equation}
There is thus a single degree of freedom left, and it can be parametrized by the value of the remaining correlator
$E_B$. Its value 
\begin{equation}\label{eq: EB}
     E_B = 2(p_\alpha p_\gamma +(1-p_\alpha)(1-p_\gamma)) -1,
 \end{equation}
can be chosen freely in the interval $[-1,1]$, and any choice of $p_\alpha$  and $p_\gamma$ fulfilling Eq.~\eqref{eq: EB} realises it.
Hence, we conclude that there exists a PTC distribution with  
$(E_A,E_B,E_C) = (0, E_B, 0)$ (and any permutation of parties), and there are infinitely many PTC strategies that achieve it. Again we can distinguish two cases. 

\begin{itemize}
\item[\textbf{2a.}] For a nonnegative  value $E_B\geq 0$ one can exchange the parties $A$ and $B$ and realize that $(E_A\geq 0,0)$ is a continuation of the case (1.a) where  $(E_A\geq E_B\geq E_C > 0)$ \\
\item[\textbf{2b.}] Similarly, for a negative value $E_B<0$ one exchanges $B$ with $C$ and notice that $(0,0, E_C<0)$ is the continuation of the case (1.b) where ($E_A> 0 > E_B \geq E_C)$.
\end{itemize}
With now combine all the cases the cases $1$ and $2$ in the following observation.\\

\noindent \textbf{Theorem 3.}\textit{
There exists a PTC distribution $P_{PTC}(a,b,c)$ with the correlator values $E_A\geq E_B \geq E_C$  (where $E_{A}= 2 P(a=1)-1$) if and only if one of the following is true
 \begin{itemize}
 \item[\textbf{1a.}]  $E_A,E_B,E_C > 0$ and  $E_C \geq E_A E_B$,
 \item[\textbf{1b.}]
      $E_A >0$; $E_B , E_C < 0$ and   $E_A |E_C| \leq |E_B |\leq \frac{E_A}{|E_C|}$,
\item[\textbf{2a.}] $E_A\geq 0$ and $E_B = E_C =0$,
\item[\textbf{2b.}] $E_A= E_B = 0$ and  $ E_C< 0$.
\end{itemize}
The PTC distribution is unique and given in Table.~\ref{tab:PTC corr}.
 In addition, in the cases \textbf{1a} and \textbf{1b} the distribution can be simulated with exactly two strategies  $(p_\alpha ,p_\beta ,p_\gamma) =\left(\frac{1}{2}(1+\sqrt{\frac{E_B E_C}{E_A}}), \frac{1}{2}(1+\sqrt{\frac{E_A E_C}{E_B}}),\frac{1}{2}(1+\sqrt{\frac{E_A E_B}{E_C}})\right)$ and $(1-p_\alpha,1-p_\beta,1-p_\gamma)$. In the cases \textbf{2a} and \textbf{2b} there are infinitely many PTC strategies (1-parameter family) that simulate the distribution.}\\

Obviously, the observation remains true for any permutation of the parties. This gives a full characterization of the set $\mathcal{L}_{PTC}$ and the multiplicity of the underlying PTC strategies.

\section{Proof of Almost-Rigidity of Parity Token Counting distributions} \label{Appendix_SemiRig}
\begin{proof}
Take any value of the outputs $a_0 b_0 c_0$ that has non-zero probability and  consider the values $\alpha_0,\beta_0, \gamma_0$ of the local variables that lead to $a_0 b_0 c_0$. Now, we define the derivative of an output with respect to a local variable as follows
\begin{equation}\label{eq:derivative}
    \nabla_{\alpha_0} b(\gamma,\alpha) = b(\gamma,\alpha) \ominus b(\gamma,\alpha_0).
\end{equation}
Here by $\ominus$ we denote the difference modulo 2, $x \ominus y = (x -y) \, \text{mod} \, 2$.

Since $a\oplus b\oplus c =1$ for any PTC distribution, we have: 
\begin{equation*}
    \nabla_{\alpha_0} a(\beta, \gamma) \oplus \nabla_{\alpha_0} b(\gamma,\alpha) \oplus \nabla_{\alpha_0} c(\alpha, \beta) = 0
\end{equation*}
The output $a(\beta, \gamma)$ does not depend on $\alpha$, therefore we have $\nabla_{\alpha_0} a(\beta, \gamma) = a(\beta, \gamma)- a(\beta, \gamma) = 0$ and 
$\nabla_{\alpha_0} b(\gamma,\alpha) = \nabla_{\alpha_0} c(\alpha, \beta)$. Here $\nabla_{\alpha_0} b(\gamma,\alpha)$ is independent of $\beta$ and $\nabla_{\alpha_0} c(\alpha, \beta)$ is independent of $\gamma$, being always equal they can only depend on $\alpha$. We can thus define 
\begin{equation}\label{eq:f_alpha_b}
     f_{\alpha_0} (\alpha) = \nabla_{\alpha_0} b(\gamma,\alpha) = \nabla_{\alpha_0} c(\alpha, \beta). 
\end{equation}
In a similar manner, one defines $f_{\beta_0}(\beta)=\nabla_{\beta_0} a(\beta,\gamma) = \nabla_{\beta_0} c(\alpha, \beta)$ and $f_{\gamma_0}(\gamma)=\nabla_{\gamma_0} a(\beta,\gamma) = \nabla_{\gamma_0} b(\alpha, \gamma)$. 

Now, following the equation \ref{eq:f_alpha_b} and the definition of the derivative, equation \ref{eq:derivative}, we have 
\begin{equation*}\begin{split}
     b(\gamma,\alpha) &=  b(\gamma,\alpha_0) \oplus f_{\alpha_0} (\alpha) \\
     b(\gamma,\alpha) &=  b(\gamma_0,\alpha) \oplus f_{\gamma_0} (\gamma)
\end{split}
\end{equation*}
Combining the two gives
\begin{equation*}
     b(\gamma,\alpha) = b(\gamma_0,\alpha_0) \oplus f_{\alpha_0} (\alpha) \oplus f_{\gamma_0} (\gamma) = \underbrace{b_0 \oplus f_{\alpha_0} (\alpha)}_{\equiv T_{\alpha}^b(\alpha)} \oplus \underbrace{f_{\gamma_0} (\gamma)}_{\equiv T_{\gamma}^b(\gamma)}
\end{equation*} 
where we defined the functions  $T_{\alpha}^b(\alpha)$ and $T_{\gamma}^b(\gamma)$. Similarly, we can define 
\begin{align*}
    & a(\beta, \gamma) =  a_0 \oplus f_{\gamma_0} (\gamma)\oplus f_{\beta_0} (\beta) = \underbrace{1 \oplus f_{\gamma_0} (\gamma)}_{\equiv T_{\gamma}^a (\gamma)} \oplus \underbrace{1\oplus a_0\oplus f_{\beta_0} (\beta)}_{\equiv T_{\beta}^a (\beta)}\\
    & c(\alpha, \beta) = c_0 \oplus f_{\beta_0} (\beta) \oplus f_{\alpha_0} (\alpha) =  \underbrace{1\oplus b_0 \oplus c_0 \oplus f_{\beta_0} (\beta)}_{\equiv T_{\beta}^c (\beta)} \oplus \underbrace{1\oplus b_0\oplus f_{\alpha_0} (\alpha)}_{\equiv T_{\alpha}^c (\alpha)} 
\end{align*}

Manifestly, defined in the above way the function ${T_{j}^i}$ fulfill \rm (ii). To prove \rm (i) simply note that 
\begin{equation*}\begin{split}
T_{\gamma}^b(\gamma) \oplus T_{\gamma}^a(\gamma) &= f_{\gamma_0}(\gamma) \oplus 1\oplus   f_{\gamma_0}(\gamma) =1
\\
T_{\alpha}^b(\alpha) \oplus T_{\alpha}^c(\alpha) &=
b_0\oplus f_{\alpha_0}(\alpha) \oplus 1 \oplus b_0 \oplus f_{\alpha_0}(\alpha) =1 
\\
T_{\beta}^a(\beta) \oplus T_{\beta}^c(\beta) &= 
1 \oplus a_0 \oplus f_{\beta_0}(\beta)\oplus 1 \oplus b_0 \oplus c_0 \oplus f_{\beta_0}(\beta) =1,
\end{split}
\end{equation*}
where we used $a_0\oplus b_0 \oplus c_0 =1$ for the last equality.

Finally, let us prove (iii) for the distributions with $P(a=1),P(b=1), P(c=1)\neq \frac{1}{2}$. In the appendix~\ref{Appendix_PRCdistributions} (Theorem 3) we proved that in this case $(E_A,E_B,E_C \neq 0)$ there are only two PTC strategies that achieve $P_{PTC}(a,b,c)$, namely $(p_\alpha, p_\beta, p_\gamma)$ and $(1-p_\alpha, 1-p_\beta, 1-p_\gamma)$. We thus know that there are function $T_{i}^j$ that satisfy (i) and (ii) and can only lead to two possibilities 
\begin{equation}
    (
\mathds{E}[T_{\alpha}^b],
\mathds{E}[T_{\beta}^c],
\mathds{E}[T_{\gamma}^a]
)=(p_\alpha, p_\beta, p_\gamma) \qquad \text{or}\qquad (
\mathds{E}[T_{\alpha}^b],
\mathds{E}[T_{\beta}^c],
\mathds{E}[T_{\gamma}^a]
)=(1-p_\alpha, 1-p_\beta, 1-p_\gamma)
\end{equation}
In the first case (iii) holds automaticaly. In the second one can define new compression functions $\bar T_{i}^j = T_{i}^j \oplus 1$ with "flipped" token directions. These function also satisfy (i,ii) and give  $\mathds{E}[\bar T_{i}^j]= 1-\mathds{E}[T_{i}^j]$. Hence, using them as compression functions gives back the original strategy $(p_\alpha, p_\beta, p_\gamma)$, and proves (iii).
\end{proof}

\section{Proof of main result: Theorem \ref{theorem_parity_qubit}}

\label{Appendix_RGB4_proof} 

Consider each source distributing an entangled two-qubit state 
\begin{equation*}
\ket{\psi} = \lambda_0\ket{01}+\lambda_1\ket{10}.
\end{equation*}
Where $\lambda_0^2 + \lambda_1^2 = 1$, and $\lambda_0 \not= \sqrt{1/2}$. Each party performs the following two-qubit joint measurement with ternary outputs $a,b,c\,{\in}\,\{\bar{0}, \bar{1}_0, \bar{1}_1 \}$.
\begin{equation*} 
  \begin{split}
       \bar{0} : \Pi_{\bar 0} &=\ketbra{00}{00} + \ketbra{11}{11} \\
       \bar{1}_0 : \Pi_{\bar {1}_{0}} &= \ketbra{\bar{1}_0}{\bar{1}_0} \\
       \bar{1}_1 : \Pi_{\bar{1}_1} &= \ketbra{\bar{1}_1}{\bar{1}_1} \\
  \end{split}
\end{equation*}
where we define, $\ket{\bar{1}_0} =  u\ket{01}+v\ket{10},~ \ket{\bar{1}_1} = v\ket{01}-u\ket{10}$, with $0<v<u<1$ real parameters that satisfy $u^2\,{+}\,v^2\,{=}\,1$. The resulting probability distribution, $P_{Q}$, is given by 
\begin{align*}
P_Q(a,b,c) = \Tr[(\Pi_a\otimes \Pi_b\otimes \Pi_c) \ketbra{\Psi}{\Psi}] ~ ~ ~ ~ ~ ~ ~ ~ a,b,c\,{\in}\,\{\bar{0}, \bar{1}_0, \bar{1}_1 \}
\end{align*}
Where in $\ket{\Psi}$, the respective Hilbert spaces are suitably ordered according to the triangle configuration.
\begin{align*}
\ket{\Psi}_{A_1 A_2 B_1 B_2 C_1 C_2} \equiv \ket{\psi}_{A_2 B_1} \otimes \ket{\psi}_{B_2 C_1} \otimes \ket{\psi}_{C_2 A_1}
\end{align*}
Introducing the notation $u_0 = -v_1 =  u$ and $v_0 =  u_1 = v$ (such that $\ket{\bar{1}_t}=u_t \ket{01} + v_t \ket{10}$) we have that:
\begin{align}
P_Q(\bar{1}_i, \bar{0},\bar{0}) & = \lambda_1^4 \lambda_0^2 u_i^2 + \lambda_0^4 \lambda_1^2 v_i^2\label{app:constraint2}\\
P_Q( \bar{1}_i, \bar{1}_j , \bar{1}_k) &=  (\lambda_1^3 u_i u_j u_k + \lambda_0^3 v_i v_j v_k)^2 \label{constraint3_appendix}\\
P_Q( \bar{1}_i, \bar{1}_j , \bar 0) &= P_Q(\bar 0,\bar 0,\bar 0)=0 \label{app:constraint3}
\end{align}
and similar relations for permutions  of the parties.  We will now prove the following theorem.\\

\textbf{Theorem 2 (generalized)} 
\textit{
For $\lambda_0^2\neq \frac{1}{2}$ the quantum distribution $P_Q(a,b,c)$ in Eqs.~(\ref{app:constraint2}-\ref{app:constraint3}) is not triangle-local (incompatible with any model of the form ~\eqref{trilocal}) if} 
\begin{equation}\label{eq: the statement appendix}
3(\lambda_1^3 u^2 v - \lambda_0^3 u v^2)^2- 3u^2 (\lambda_1^6+\lambda_0^6) + 2(\lambda_0^6 + (\lambda_1^3 u^3+\lambda_0^3 v^3)^2) + \lambda_1^6 +(\lambda_1^3 v^3-\lambda_0^3 u^3)^2 < 0.
\end{equation}\\

Notably the condition of nonlocality in Eq.~\eqref{app:constraint2} is exacly similar to the one obtain in Ref. \cite{Renou_2019} for the so-called family of  RGB4 distributions (excepts for the case with $\lambda_0^2=1/2$). \\

Note also that the formulation of the Theorem given in the main text is just a restriction of the above statement to the particular case $\lambda_0^2 = 1/3$. Indeed, for this value the bound in Eq.~\eqref{eq: the statement appendix} becomes particularly simple $u^2 >  \frac{2}{3}$. \\

 We will now prove the theorem by contradiction assuming the existence of a trilocal model reproducing $P_Q(a,b,c)$,  we identify conditions that this model should satisfy, leading to  a contradiction for certain choices of the measurement parameter $u$.\\

Notice that the measurement can be performed in two steps. In the first step, the parties measure the operators in the basis $\bar{0}\,{ : }\,\ketbra{00}{00}+\ketbra{11}{11}$, $\bar{1}\,{ : }\,\ketbra{01}{01}+\ketbra{10}{10}$, where $\bar{n}$ is the parity of received tokens, obtaining the distribution $P_{parity}$. Second, whenever $\bar{1}$ is obtained, the party performs an additional measurement given by $\bar{1}_0$, $\bar{1}_1$, which are the only entangled measurement operators in this setting. Since the measurement operators corresponding to the binned outcomes $\{\bar{0}, \bar{1}\}$, are diagonal in the basis of token numbers, substituting the quantum sources with the sources preparing separable states $\lambda_0^2 \ketbra{01}{01}+\lambda_1^2\ketbra{10}{10}$ has no effect on $P_{parity}$. This is equivalent to the classical parity token-counting strategy with the sources delivering their token to the left with probability $\lambda_0^2$ and to the right with probability $\lambda_1^2$. As a result, any potential classical strategy that simulates $P_\text{Q}$ imposes a classical strategy for $P_{parity}$, which must be the classical parity token-counting strategy, due to the almost-rigidity property.

\begin{figure}[H]
    \centering
    \includegraphics[scale=0.4]{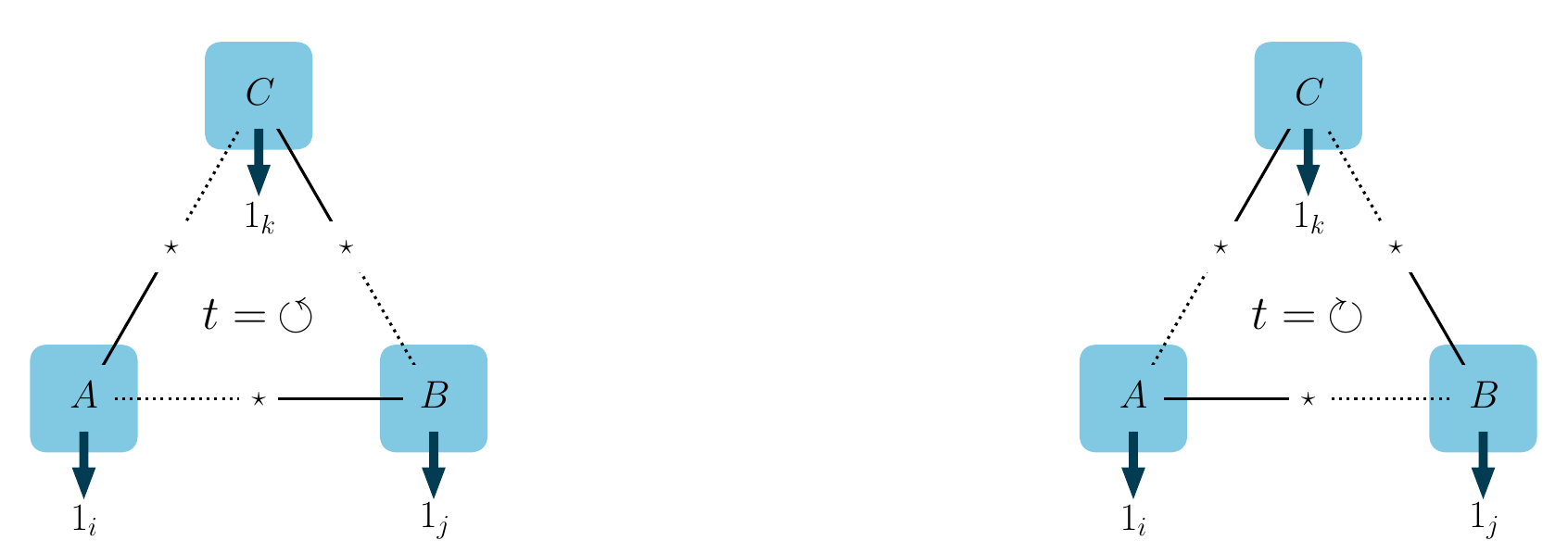}
    \caption{In uniform token-counting strategy over the triangle, when all parties obtain $\bar{1}$, the token paths are undetermined between $t\,{=}\circlearrowleft$ and $t\,{=}\circlearrowright$. In a quantum scenario the system could be in a superposition state of $t\,{=}\circlearrowleft$ and $t\,{=}\circlearrowright$. However in a classical scenario the underlying trajectory of tokens is either $t\,{=}\circlearrowleft$ or $t\,{=}\circlearrowright$.}
    \label{fig:triangle_clk_unticlk}
\end{figure}

Therefore, any potential classical simulating strategy for $P_\text{Q}$, involves a new classical hidden variable which shows the movement of each individual tokens. In particular, considering the case when all parties output $\bar{1}$, the source's tokens must be transmitted in either the clockwise ($t\,{=}\circlearrowright$) or in the anti-clockwise ($t\,{=}\circlearrowleft$) direction, as illustrated in Figure~\ref{fig:triangle_clk_unticlk}. Therefore, if a trilocal model existed, one should be able to define the following joint probability distribution for $ t \in \{ \circlearrowright , \circlearrowleft\}$
\begin{equation}
   q(i,j,k,t) = \text{Pr}(a=\bar{1}_i,\ b=\bar{1}_j,\ c=\bar{1}_k,\ t \ | \ \mathrm{all~1}) =
   \frac{\text{Pr}(a=\bar{1}_i,\ b=\bar{1}_j,\ c=\bar{1}_k,\ t)}{\lambda_1^6 + \lambda_0^6}
\end{equation}
where "all~1" indicates the event of all parties outputting one token, and the last equality is due to the fact that $\text{Pr}(\mathrm{all~1}) = {\lambda_1^6 + \lambda_0^6}$. We know that this distribution must satisfy
\begin{equation}
    q(i,j,k)= \sum_{t =  \circlearrowright , \circlearrowleft} q(i,j,k,t) = \frac{P_Q( \bar{1}_i, \bar{1}_j , \bar{1}_k)}{\lambda_1^6 + \lambda_0^6} = \frac{( \lambda_1^3 u_i u_j u_k +  \lambda_0^3 v_i v_j v_k )^2}{\lambda_1^6 + \lambda_0^6} . \label{marg_q_ijk_appendix}
\end{equation}

Next, we explicitly compute some of the marginals of $q(i,j,k,t)$ using the fact that the underlying model abides to the triangle network. To do so we use the notation $\alpha \mapsto B$ to mean that the source $\alpha$ sent its token to the party $B$. The network stature implies that Alice's output must be independent of the $\alpha$ source, in particular
\begin{align*}
    \text{Pr}(a=\bar{1}_i,\gamma \mapsto A, \beta \mapsto C| \alpha \mapsto B) &= \text{Pr}(a=\bar{1}_i,\gamma \mapsto A, \beta \mapsto C|\alpha \mapsto C) \qquad \Leftrightarrow \\
   \frac{\text{Pr}(\alpha \mapsto C) }{\text{Pr}(\alpha \mapsto B)}  \text{Pr}(a=\bar{1}_i, \alpha \mapsto B,\gamma \mapsto A, \beta \mapsto C) &= \text{Pr}(a=\bar{1}_i, \alpha \mapsto C,\gamma \mapsto A, \beta \mapsto C).
\end{align*}
Furthermore, note that here 
\begin{align*} 
\text{Pr}(a=\bar{1}_i, \alpha \mapsto B,\gamma \mapsto A, \beta \mapsto C)
&= \text{Pr}(a=\bar{1}_i,t=\circlearrowright) 
= \sum_{j,k}  \text{Pr}(a=\bar{1}_i,b=\bar{1}_j,c=\bar{1}_k,t=\circlearrowright) \\
&=\text{Pr}(\mathrm{all~1})\sum_{jk} q(i,j,k,t=\circlearrowright),
\end{align*}
which allows us to rewrite the last equation as
\begin{equation}
    \frac{\text{Pr}(\alpha \mapsto C) }{\text{Pr}(\alpha \mapsto B)}  \,\text{Pr}(\mathrm{all~1})\sum_{jk} q(i,j,k,t=\circlearrowright) = \text{Pr}(a=\bar{1}_i, \alpha \mapsto C,\gamma \mapsto A, \beta \mapsto C).
\end{equation}
In the same way for $t=\circlearrowleft$ one obtains the relation
\begin{equation}\label{eq_app:const_qi_a}
\frac{\text{Pr}(\alpha \mapsto B)} {\text{Pr}(\alpha \mapsto C)} P(\mathrm{all~1})\sum_{jk} q(i,j,k,t=\circlearrowleft) = \text{Pr}(a=\bar{1}_i, \alpha \mapsto B,\gamma \mapsto B, \beta \mapsto A).
\end{equation}
Finally, using $\text{Pr}(\alpha \mapsto B) = \lambda_1^2 , \text{Pr}(\alpha \mapsto C) = \lambda_0^2, \text{Pr}(\mathrm{all~1}) = \lambda_1^6 + \lambda_0^6$, and 
$\text{Pr}(a=\bar{1}_i, \alpha \mapsto C,\gamma \mapsto A, \beta \mapsto C) + \text{Pr}(a=\bar{1}_i, \alpha \mapsto B,\gamma \mapsto B, \beta \mapsto A) = \text{Pr}(1_i,0,0)$, we obtain the desired constraint on the marginals of $q(i,j,k,t)$
\begin{equation}
\frac{\lambda_0^2}{\lambda_1^2}\underbrace{\sum_{j,k} q(i,j,k,t=\circlearrowright) }_{q(i,t=\circlearrowright)}+ \frac{\lambda_1^2}{\lambda_0^2}\underbrace{\sum_{j,k} q(i,j,k,t=\circlearrowleft)}_{q(i,t=\circlearrowleft)}  = \frac{\text{Pr}(1_i,0,0)}{\text{Pr}(\mathrm{all~1}) } =\frac{\lambda_1^4 \lambda_0^2 u_i^2 + \lambda_0^4 \lambda_1^2 v_i^2}{\lambda_1^6+\lambda_0^6} ,\label{marg_q_it_appendix}.
\end{equation}
By symmetry the same constraints hold for the marginal distributions of the other parties $q(j,t)$ and $q(k,t)$.
\\

Assuming the existence of a trilocal model, one should be able to define a distribution $q(i,j,k,t)$ that is consistent with the equality constraints of Eqs.
(\ref{marg_q_ijk_appendix}, \ref{marg_q_it_appendix}). This satisfying marginal constraints problem is a simple linear program (LP), which can be solved efficiently. In the next section \ref{appendix_lemma_noQ} we demonstrate analytically that this LP has no solution for some specific choices of the operators $\{\bar{1}_0,\bar{1}_1\}$, precisely when the parameters $\lambda_1$ and $u$ fulfill the condition given in Eq.~\eqref{eq: the statement appendix}, thereby prooving the theorem.\\

A key distinction between classical and quantum token-counting strategies is that in the quantum approach, when all parties have received one token, the global state is in a coherent superposition of tokens cycling clockwise and anti-clockwise. In contrast, in the classical case the whole system is either in state of tokens cycling clockwise or anti-clockwise with some probability
\begin{align*}
\mathrm{Classical~:~}&\rho = \lambda_1^2 \ketbra{\circlearrowright}{\circlearrowright}+\lambda_0^2\ketbra{\circlearrowleft}{\circlearrowleft}\\
\mathrm{Quantum~:~}& \rho = \Big(\lambda_1 \ket{\circlearrowright}+\lambda_0\ket{\circlearrowleft}\Big)\Big(\lambda_1 \bra{\circlearrowright}+\lambda_0\bra{\circlearrowleft}\Big)
\end{align*}
It is the global coherence between the superposed branches of all tokens going left and all tokens goings right which is ultimately responsible for the nonlocality of the quantum distribution.

\subsection{There is no $q(i,j,k,t)$ compatible with the marginals}
\label{appendix_lemma_noQ}

Let us present the probability distribution $\tilde{q}$ resulting from symmetrizing $q$ across the variables $i,j,k$:
\begin{equation*}
\tilde{q}(i,j,k,t)=\frac{1}{6} \left(q(i,j,k,t)+q(j,k,i,t)+q(k,i,j,t)+q(i,k,j,t)+q(k,j,i,t)+q(j,i,k,t)\right).
\end{equation*}

Evidently, $\tilde{q}$ satisfies the same constraints given in Eqs.~(\ref{marg_q_ijk_appendix}, \ref{marg_q_it_appendix}) as $q$, since they are invariant under the permutation of indices. Let us define $\xi_{ijk}:=\tilde{q}(i,j,k,t=\circlearrowright)-\tilde{q}(i,j,k,t=\circlearrowleft)$. Therefore we have
\begin{align*}
\tilde{q}(i,j,k,t=\circlearrowright)=\frac{1}{2} \left( \tilde{q}(i,j,k) + \xi_{ijk} \right), ~~~~~~~~
\tilde{q}(i,j,k,t=\circlearrowleft)=\frac{1}{2} \left( \tilde{q}(i,j,k) - \xi_{ijk} \right).
\end{align*}
For simplicity of notation, we introduce
$\xi_0:=\xi_{000}$, $\xi_{1}:=\xi_{100}=\xi_{010}=\xi_{001}$, $\xi_2:=\xi_{110}=\xi_{101}=\xi_{011}$ and $\xi_{3}:=\xi_{111}$. 

With the help of Eq.~\eqref{marg_q_ijk_appendix}  we can now write 
\begin{equation}
    \tilde q(k) = \sum_{i,j} \tilde q(i,j,k) = \sum_{i,j}  \frac{(\lambda_1^3 u_iu_ju_k + \lambda_0^3 v_iv_jv_k)^2 }{\lambda_1^6 + \lambda_0^6} = \frac{\lambda_1^6u_k^2+\lambda_0^6v_k^2}{\lambda_1^6+\lambda_0^6},
\end{equation}
and
\begin{align}
\label{eq: tilde q whatever 1}
\tilde{q}(k,t=\circlearrowright)&=\frac{1}{2} \sum_{i,j} \left( \tilde q(i,j,k)  + \xi_{ijk} \right)=\frac{1}{2} \left( \frac{\lambda_1^6u_k^2+\lambda_0^6v_k^2}{\lambda_1^6+\lambda_0^6} + \sum_{ij} \xi_{ijk} \right)\\
\label{eq: tilde q whatever 2}
\tilde{q}(k,t=\circlearrowleft)&=\frac{1}{2} \sum_{i,j} \left( \tilde q(i,j,k)  - \xi_{ijk} \right) = \frac{1}{2} \left( \frac{\lambda_1^6u_k^2+\lambda_0^6v_k^2}{\lambda_1^6+\lambda_0^6} - \sum_{ij} \xi_{ijk} \right).
\end{align}

Now let use these expression in Eq.~\eqref{marg_q_it_appendix}, we have 
\begin{align}
 \frac{\lambda_1^4 \lambda_0^2 u_i^2 + \lambda_0^4 \lambda_1^2 v_i^2}{\lambda_1^6+\lambda_0^6} &=    \frac{\lambda_0^2}{\lambda_1^2}\tilde{q}(k,t=\circlearrowright) + \frac{\lambda_1^2}{\lambda_0^2}\tilde{q}(k,t=\circlearrowleft) 
\\
&= 
  \frac{\lambda_0^2}{\lambda_1^2} \frac{1}{2} \left( \frac{\lambda_1^6u_k^2+\lambda_0^6v_k^2}{\lambda_1^6+\lambda_0^6} + \sum_{ij} \xi_{ijk} \right) + \frac{\lambda_1^2}{\lambda_0^2} \frac{1}{2} \left( \frac{\lambda_1^6u_k^2+\lambda_0^6v_k^2}{\lambda_1^6+\lambda_0^6} - \sum_{ij} \xi_{ijk} \right)\\
  &= \frac{1}{2}\left(\frac{\lambda_0^2}{\lambda_1^2}+ \frac{\lambda_1^2}{\lambda_0^2}\right)\frac{\lambda_1^6u_k^2+\lambda_0^6v_k^2}{\lambda_1^6+\lambda_0^6} + \frac{1}{2}\left(\frac{\lambda_0^2}{\lambda_1^2}- \frac{\lambda_1^2}{\lambda_0^2}\right)\sum_{ij} \xi_{ijk},
\end{align}
which simplifies to 
\begin{equation}\label{eq: D5}
    \sum_{ij} \xi_{ijk}  = 
    \frac{(\lambda_1^6 \lambda_0^4 - \lambda_1^{10})u_k^2 + (\lambda_0^6 \lambda_1^4 - \lambda_0^{10})v_k^2}
    {(\lambda_0^4-\lambda_1^4)(\lambda_1^6+\lambda_0^6)}=
    \frac{\lambda_1^6u_k^2-\lambda_0^6v_k^2}{\lambda_1^6+\lambda_0^6}.
\end{equation}
Where we used $\lambda_1^2+\lambda_0^2=1$ and the fact that $\lambda_1^2 \not= \lambda_0^2$. Eq.~\eqref{eq: D5} gives two different constraints for $k=0$ and $k=1$. Using $\sum_{ij} \xi_{ij0}=\xi_0+ 2 \xi_1+\xi_2$ and $\sum_{ij} \xi_{ij1}=\xi_1+ 2 \xi_2+\xi_3$ we can put them in the form
\begin{align}\label{eq: d8}
\xi_0&=\frac{\lambda_1^6u^2-\lambda_0^6v^2}{\lambda_1^6+\lambda_0^6} - 2 \xi_1 - \xi_2=u^2-\frac{\lambda_0^6}{\lambda_1^6+\lambda_0^6} - 2 \xi_1 - \xi_2,\\\label{eq: d9}
\xi_3&=\frac{\lambda_1^6 v^2 - \lambda_0^6 u^2}{\lambda_1^6+\lambda_0^6} - \xi_1 - 2 \xi_2 = \frac{\lambda_1^6}{\lambda_1^6+\lambda_0^6} - u^2 -\xi_1 - 2 \xi_2 .
\end{align}

Finally, we use the there positivity constraints that we rewrite with the help of Eq.~(\ref{marg_q_ijk_appendix})
\begin{align}
0 &\leq \tilde{q}(0,0,0,t=\circlearrowleft)= \frac{1}{2}\left(\frac{(\lambda_1^3 u^3+\lambda_0^3 v^3)^2}{(\lambda_1^6+\lambda_0^6)} - \xi_0\right), \\
0 &\leq \tilde{q}(1,1,1,t=\circlearrowright)=  \frac{1}{2}\left(\frac{(\lambda_1^3 v^3-\lambda_0^3 u^3)^2}{(\lambda_1^6+\lambda_0^6)} +\xi_3\right)\\ \label{eq: positivity last}
0&\leq \tilde{q}(0,0,1,t=\circlearrowleft)= \frac{1}{2}\left( \frac{(\lambda_1^3 u^2 v - \lambda_0^3 u v^2)^2}{(\lambda_1^6+\lambda_0^6)} - \xi_1\right).
\end{align}
When combined with Eqs.~(\ref{eq: d8},\ref{eq: d9}) the first two positivity constrains become 
\begin{align}
\xi_2 \geq u^2-\frac{\lambda_0^6 + (\lambda_1^3 u^3+\lambda_0^3 v^3)^2}{ (\lambda_1^6+\lambda_0^6)}  - 2 \xi_1, \quad
\xi_2 \leq \frac{\lambda_1^6 +(\lambda_1^3 v^3-\lambda_0^3 u^3)^2}{2(\lambda_1^6+\lambda_0^6)} - \frac{u^2}{2} - \frac{\xi_1}{2},
\end{align}
and lead to a lower bound on $\xi_1$ when cobined
\begin{align}
\xi_1 \geq u^2 - \frac{2(\lambda_0^6 + (\lambda_1^3 u^3+\lambda_0^3 v^3)^2) + \lambda_1^6 +(\lambda_1^3 v^3-\lambda_0^3 u^3)^2}{3 (\lambda_1^6+\lambda_0^6)}.
\end{align}
In addition, Eq.~\eqref{eq: positivity last} gives an upper bound on $\xi_1$, and two constraints need to hold together. Therefore the upper bound on $\xi$ must exceed the lower bound which implies the following final bound that must be satisfies for the local model to exist
\begin{align}
&\frac{(\lambda_1^3 u^2 v - \lambda_0^3 u v^2)^2}{(\lambda_1^6+\lambda_0^6)}  \geq u^2 - \frac{2(\lambda_0^6 + (\lambda_1^3 u^3+\lambda_0^3 v^3)^2) + \lambda_1^6 +(\lambda_1^3 v^3-\lambda_0^3 u^3)^2}{3 (\lambda_1^6+\lambda_0^6)} \qquad
\implies \\
&3(\lambda_1^3 u^2 v - \lambda_0^3 u v^2)^2- 3u^2 (\lambda_1^6+\lambda_0^6) + 2(\lambda_0^6 + (\lambda_1^3 u^3+\lambda_0^3 v^3)^2) + \lambda_1^6 +(\lambda_1^3 v^3-\lambda_0^3 u^3)^2 \geq 0.\label{ineq}
 \end{align}
This inequality defines the region $u > u_\mathrm{max}(\lambda_1)$ for which the distribution $P_Q$ does not admit a local model, see Fig.~\ref{function}. Note that by symmetry, exchanging $\lambda_1$ with $\lambda_0$ and $u$ with $v$ does not change the distributions, which define another nonlocality region  $u < \sqrt{1-(u_\mathrm{max}(\sqrt{1-\lambda_1^2}))^2}$ also depicted in Fig.~\ref{function}. In this figure, we can see that we recover the same non-locality range of the ref. \cite{Renou_2019} for the coarse grainded RGB4 distribution (except the measure-zero case of $\lambda_1 = \lambda_0 =1/\sqrt{2}$).

\begin{figure}
\centering
\includegraphics[width=0.6\columnwidth]{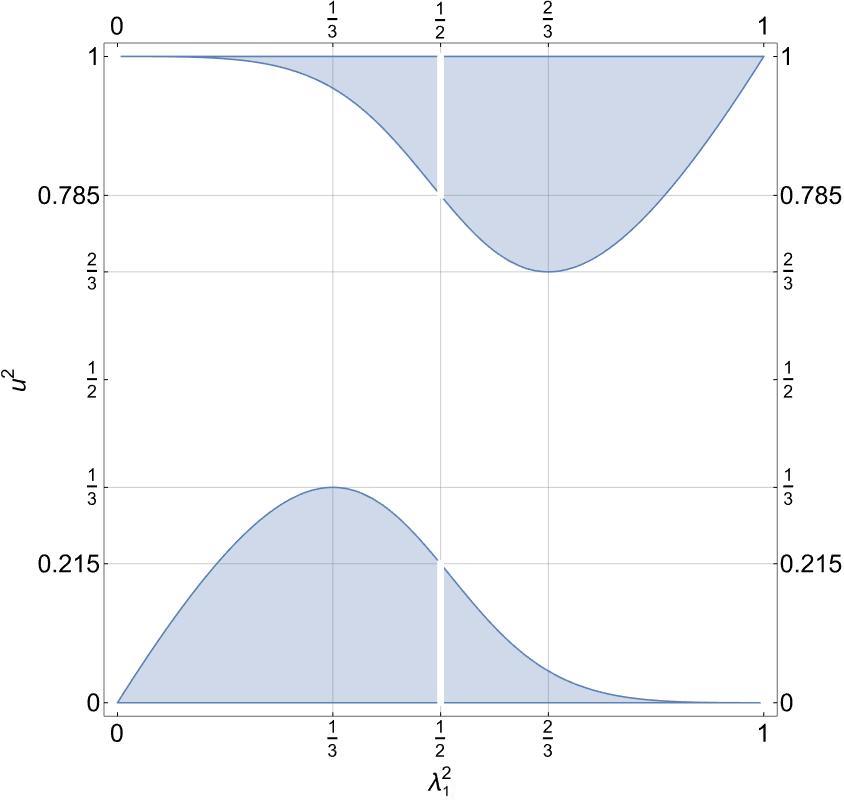}
\caption{The shaded region depicts the regime of parameters $u^2$ and $\lambda_1^2$ for which we prove the distribution $P_Q(a,b,c)$ to be nonlocal. Note that the line $\lambda_1^2=\frac{1}{2}$ is excluded, as our proof does not apply to this case. In the plot we also considered $u\leq v$.
} \label{figure_elisa}
\label{function}
\end{figure}

\section{Proof of non-locality by inflation}
\label{appendix_inflation}
To prove the non-locality of several variations of the coarse grained Fritz distribution and to corroborate the results of the Neural Network a Web inflation of the triangle scenario was used \cite{wolfe2019inflation}. Each source was inflated up to two copies.

\begin{center}
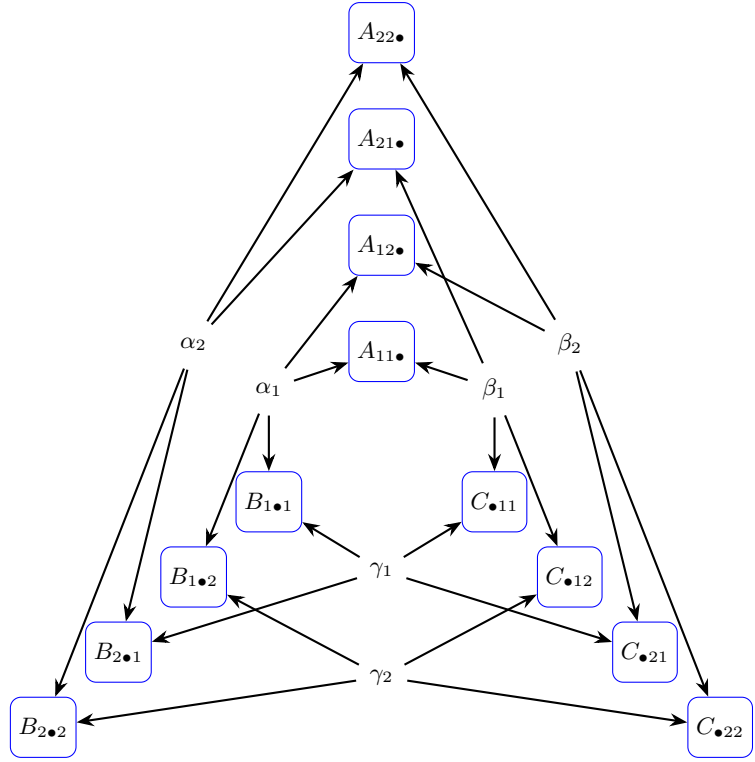
\begin{figure}[h!]
\begin{tikzpicture}
    \node[shape =  rectangle, minimum size=0.8cm, rounded corners ,draw=blue] (A11) at (3,3) {$A_{11 \bullet}$};
    \node[shape =  rectangle, minimum size=0.8cm, rounded corners ,draw=blue] (A12) at (3,3+1.414) {$A_{12 \bullet}$};
    \node[shape =  rectangle, minimum size=0.8cm, rounded corners ,draw=blue] (A21) at (3,3+1.414*2) {$A_{21 \bullet}$};
    \node[shape =  rectangle, minimum size=0.8cm, rounded corners ,draw=blue] (A22) at (3,3+1.414*3) {$A_{22 \bullet}$};
    \node[shape =  rectangle, minimum size=0.8cm, rounded corners ,draw=blue] (C11) at (4.5,1) {$C_{\bullet 11}$};
    \node[shape =  rectangle, minimum size=0.8cm, rounded corners ,draw=blue] (C12) at (4.5+1,1-1) {$C_{\bullet 12}$};
    \node[shape =  rectangle, minimum size=0.8cm, rounded corners ,draw=blue] (C21) at (4.5+1*2,1-1*2) {$C_{\bullet 21}$};
    \node[shape =  rectangle, minimum size=0.8cm, rounded corners ,draw=blue] (C22) at (4.5+1*3,1-1*3) {$C_{\bullet 22}$};
    \node[shape =  rectangle, minimum size=0.8cm, rounded corners ,draw=blue] (B11) at (1.5,1) {$B_{1\bullet 1}$};
    \node[shape =  rectangle, minimum size=0.8cm, rounded corners ,draw=blue] (B12) at (1.5-1,1-1) {$B_{1\bullet 2}$};
    \node[shape =  rectangle, minimum size=0.8cm, rounded corners ,draw=blue] (B21) at (1.5-1*2,1-1*2) {$B_{2\bullet 1}$};
    \node[shape =  rectangle, minimum size=0.8cm, rounded corners ,draw=blue] (B22) at (1.5-1*3,1-1*3) {$B_{2\bullet 2}$};
    \node[shape=circle,draw=white] (U31) at (3,1.5-1.414) {$\gamma_{1}$};
    \node[shape=circle,draw=white] (U32) at (3,1.5-1.414*2) {$\gamma_{2}$};
    \node[shape=circle,draw=white] (U11) at (2.5-1,1.5+1) {$\alpha_{1}$};
    \node[shape=circle,draw=white] (U12) at (2.5-1*2,1.1+1*2) {$\alpha_{2}$};
    \node[shape=circle,draw=white] (U21) at (3.5+1,1.5+1) {$\beta_{1}$} ;
    \node[shape=circle,draw=white] (U22) at (3.5+1*2,1.1+1*2) {$\beta_{2}$} ;
\begin{scope}[>={Stealth[black]},
              every edge/.style={draw=black,thick}]
              
              \path [->] (U11) edge node {} (B11);
              \path [->] (U11) edge node {} (B12);
              \path [->] (U11) edge node {} (A11);
              \path [->] (U11) edge node {} (A12);
              \path [->] (U12) edge node {} (B21);
              \path [->] (U12) edge node {} (B22);
              \path [->] (U12) edge node {} (A21);
              \path [->] (U12) edge node {} (A22);              
              \path [->] (U31) edge node {} (B11);
              \path [->] (U31) edge node {} (B21);
              \path [->] (U31) edge node {} (C11);
              \path [->] (U31) edge node {} (C21);
              \path [->] (U32) edge node {} (B12);
              \path [->] (U32) edge node {} (B22);
              \path [->] (U32) edge node {} (C12);
              \path [->] (U32) edge node {} (C22);
              \path [->] (U21) edge node {} (C11);
              \path [->] (U21) edge node {} (C12);
              \path [->] (U21) edge node {} (A11);
              \path [->] (U21) edge node {} (A21);
              \path [->] (U22) edge node {} (C21);
              \path [->] (U22) edge node {} (C22);
              \path [->] (U22) edge node {} (A12);
              \path [->] (U22) edge node {} (A22);
\end{scope}
\end{tikzpicture}
\caption{\textbf{The Inflated Triangle Scenario} where $\alpha$, $\beta$ and $\gamma$ are the sources and the rest are observable variables. The inflation order is set to be 2. The copies are created such that $A_{\alpha=x,\beta=y,\gamma=z}=A_{xyz}$ where $x,y,z\in\{1,2\}$ and the symbol '$\bullet$' means that the variable in question does not have a dependency on that particular source.}
\end{figure}
\end{center}

The following two maximal expressible sets were used to run the inflation linear program. Due to the factorization of these probabilities, the inequalities obtained from the dual solution would contain the corresponding quadratic and cubic elements.
\begin{align} 
P\left( A_{11 \bullet}A_{22 \bullet} B_{1\bullet 1} B_{2\bullet 2} C_{\bullet 11} C_{\bullet 22}\right)
&=P\left(A_{11 \bullet}B_{1\bullet 1}C_{\bullet 11}\right)P\left(A_{22 \bullet}B_{2\bullet 2} C_{\bullet 22}\right),\label{D:1}\\ 
P\left(A_{12 \bullet}B_{2\bullet 1}C_{\bullet 12}\right)
&=P(A_{12 \bullet})P(B_{2\bullet 1})P(C_{\bullet 12}).\label{D:11}
\end{align}

The inflation linear program showed a varying sensitivity for different combinations of Alice's and Bob's outputs in detecting non-locality. Therefore, in some cases using only \eqref{D:1} was sufficient. Below is a table that presents the results of the inflation linear program for some coarse grained Fritz distributions.
\begin{align}
\label{D:4}
&\begin{aligned}
    &\;\text{\underline{$A$}}\\[-2pt]
    \textbf{0}\,[\!&\begin{array}{l}
         00  \\
         01  \\
    \end{array} \\[-3.5pt]\begin{array}{l}
         \textbf{1}\,\text{\footnotesize[}  \\
         \textbf{2}\,\text{\footnotesize[}
    \end{array}\!\!\! &\begin{array}{c}
          10 \\
          11
    \end{array}
\end{aligned}\;\;\;\;
\begin{aligned}
    &\;\text{\underline{$B$}}\\[-2pt]
    \begin{array}{l}
         \textbf{0}\,\text{\footnotesize[}  \\
         \textbf{1}\,\text{\footnotesize[}
    \end{array}\!\!\! &\begin{array}{c}
          00 \\
          01
    \end{array}\\[-3.5pt]
    \textbf{2}\,[\!&\begin{array}{l}
         10  \\
         11  \\
    \end{array}
\end{aligned}\;\;\;\;
\begin{aligned}
    &\;\text{\underline{$C$}}\\[-2pt]
    \textbf{0}\,[\!&\begin{array}{l}
         00  \\
         01  \\
    \end{array} \\[-3.5pt]\begin{array}{l}
         \textbf{1}\,\text{\footnotesize[}  \\
         \textbf{2}\,\text{\footnotesize[}
    \end{array}\!\!\! &\begin{array}{c}
          10 \\
          11
    \end{array}
\end{aligned}\;\;\;\;
\begin{aligned}\;\;\;\;\;\;\;\;
      &\text{\underline{Noise Tolerance}} \\
       &\\[-0.54pt]
       &\;\;\;\;\;\;\;\, \approx0.87\\[-3.54pt]
       &\\[-2.54pt]
       &\\[-1pt]
\end{aligned}\;\;\;\;\;\;\;\;&
\begin{aligned}
    &\;\text{\underline{$A$}}\\[-2pt]
    \textbf{0}\,[\!&\begin{array}{l}
         00  \\
         01  \\
    \end{array} \\[-3.5pt]\begin{array}{l}
         \textbf{1}\,\text{\footnotesize[}  \\
         \textbf{2}\,\text{\footnotesize[}
    \end{array}\!\!\! &\begin{array}{c}
          10 \\
          11
    \end{array}
\end{aligned}\;\;\;\;
\begin{aligned}
    &\;\text{\underline{$B$}}\\[-2pt]
    \begin{array}{l}
         \textbf{0}\,\text{\footnotesize[}  \\
         \textbf{1}\,\text{\footnotesize[}
    \end{array}\!\!\! &\begin{array}{c}
          00 \\
          01
    \end{array}\\[-3.5pt]
    \textbf{2}\,[\!&\begin{array}{l}
         10  \\
         11  \\
    \end{array}
\end{aligned}\;\;\;\;
\begin{aligned}
    &\;\text{\underline{$C$}}\\[-2pt]
    \textbf{0} \bigg[\hspace{-0.25em}\text{\small-}\!&\begin{array}{l}
         00  \\
         01  \\
         10  \\
    \end{array} \\[-3.5pt] \textbf{1}\text{\small[}\!&\begin{array}{c}
          11 \\
    \end{array}
\end{aligned}\;\;\;\;
\begin{aligned}\;\;\;\;\;\;\;\;
      &\text{\underline{Noise Tolerance}} \\
       &\\[-0.54pt]
       &\;\;\;\;\;\;\;\,\approx 0.87\\[-3.54pt]
       &\\[-2.54pt]
       &\\[-1pt]
\end{aligned}\\[5.5pt]
\label{D:3}
&
\begin{aligned}
    &\;\text{\underline{$A$}}\\[-2pt]
    \textbf{0}\,[\!&\begin{array}{l}
         00  \\
         01  \\
    \end{array} \\[-3.5pt]\begin{array}{l}
         \textbf{1}\,\text{\footnotesize[}  \\
         \textbf{2}\,\text{\footnotesize[}
    \end{array}\!\!\! &\begin{array}{c}
          10 \\
          11
    \end{array}
\end{aligned}\;\;\;\;
\begin{aligned}
    &\;\text{\underline{$B$}}\\[-2pt]
    \textbf{0}\,[\!&\begin{array}{l}
         00  \\
         01  \\
    \end{array} \\[-3.5pt]\begin{array}{l}
         \textbf{1}\,\text{\footnotesize[}  \\
         \textbf{2}\,\text{\footnotesize[}
    \end{array}\!\!\! &\begin{array}{c}
          10 \\
          11
    \end{array}
\end{aligned}\;\;\;\;
\begin{aligned}
    &\;\text{\underline{$C$}}\\[-2pt]
    \textbf{0}\,[\!&\begin{array}{l}
         00  \\
         01  \\
    \end{array} \\[-3.5pt]\begin{array}{l}
         \textbf{1}\,\text{\footnotesize[}  \\
         \textbf{2}\,\text{\footnotesize[}
    \end{array}\!\!\! &\begin{array}{c}
          10 \\
          11
    \end{array}
\end{aligned}\;\;\;\;
\begin{aligned}\;\;\;\;\;\;\;\;
      &\text{\underline{Noise Tolerance}} \\
       &\\[-0.54pt]
       &\;\;\;\;\;\;\;\, \approx0.81\\[-3.54pt]
       &\\[-2.54pt]
       &\\[-1pt]
\end{aligned}
\;\;\;\;\;\;\;\;&
\begin{aligned}
    &\;\text{\underline{$A$}}\\[-2pt]
    \textbf{0}\,[\!&\begin{array}{l}
         00  \\
         01  \\
    \end{array} \\[-3.5pt]\begin{array}{l}
         \textbf{1}\,\text{\footnotesize[}  \\
         \textbf{2}\,\text{\footnotesize[}
    \end{array}\!\!\! &\begin{array}{c}
          10 \\
          11
    \end{array}
\end{aligned}\;\;\;\;
\begin{aligned}
    &\;\text{\underline{$B$}}\\[-2pt]
    \textbf{0}\,[\!&\begin{array}{l}
         00  \\
         01  \\
    \end{array} \\[-3.5pt]\begin{array}{l}
         \textbf{1}\,\text{\footnotesize[}  \\
         \textbf{2}\,\text{\footnotesize[}
    \end{array}\!\!\! &\begin{array}{c}
          10 \\
          11
    \end{array}
\end{aligned}\;\;\;\;
\begin{aligned}
    &\;\text{\underline{$C$}}\\[-2pt]
    \textbf{0} \bigg[\hspace{-0.25em}\text{\small-}\!&\begin{array}{l}
         00  \\
         01  \\
         10  \\
    \end{array} \\[-3.5pt] \textbf{1}\text{\small[}\!&\begin{array}{c}
          11 \\
    \end{array}
\end{aligned}\;\;\;\;
\begin{aligned}\;\;\;\;\;\;\;\;
      &\text{\underline{Noise Tolerance}} \\
       &\\[-0.54pt]
       &\;\;\;\;\;\;\;\, \approx0.97\\[-3.54pt]
       &\\[-2.54pt]
       &\\[-1pt]
\end{aligned}\\[5.5pt]
\label{D:2}
&
\begin{aligned}
    &\;\text{\underline{$A$}}\\[-2pt]
    \textbf{0}\,[\!&\begin{array}{l}
         00  \\
         01  \\
    \end{array} \\[-3.5pt]\begin{array}{l}
         \textbf{1}\,\text{\footnotesize[}  \\
         \textbf{2}\,\text{\footnotesize[}
    \end{array}\!\!\! &\begin{array}{c}
          10 \\
          11
    \end{array}
\end{aligned}\;\;\;\;
\begin{aligned}
    &\;\text{\underline{$B$}}\\[-2pt]
         \textbf{0}\,\text{\footnotesize[}\! &\begin{array}{c}
          00 
    \end{array}\\[-3.5pt]
    \textbf{1}\,[\!&\begin{array}{l}
         01  \\
         10  \\
    \end{array}\\[-3.5pt]
         \textbf{2}\,\text{\footnotesize[}\! &\begin{array}{c}
          11 
         \end{array}
\end{aligned}\;\;\;\;
\begin{aligned}
    &\;\text{\underline{$C$}}\\[-2pt]
    \textbf{0}\,[\!&\begin{array}{l}
         00  \\
         01  \\
    \end{array} \\[-3.5pt]\begin{array}{l}
         \textbf{1}\,\text{\footnotesize[}  \\
         \textbf{2}\,\text{\footnotesize[}
    \end{array}\!\!\! &\begin{array}{c}
          10 \\
          11
    \end{array}
\end{aligned}\;\;\;\;
\begin{aligned}\;\;\;\;\;\;\;\;
      &\text{\underline{Noise Tolerance}} \\
       &\\[-0.54pt]
       &\;\;\;\;\;\;\;\, \approx0.85\\[-3.54pt]
       &\\[-2.54pt]
       &\\[-1pt]
\end{aligned}
\;\;\;\;\;\;\;\;&
\begin{aligned}
    &\;\text{\underline{$A$}}\\[-2pt]
    \textbf{0}\,[\!&\begin{array}{l}
         00  \\
         01  \\
    \end{array} \\[-3.5pt]\begin{array}{l}
         \textbf{1}\,\text{\footnotesize[}  \\
         \textbf{2}\,\text{\footnotesize[}
    \end{array}\!\!\! &\begin{array}{c}
          10 \\
          11
    \end{array}
\end{aligned}\;\;\;\;
\begin{aligned}
    &\;\text{\underline{$B$}}\\[-2pt]
         \textbf{0}\,\text{\footnotesize[}\! &\begin{array}{c}
          00 
    \end{array}\\[-3.5pt]
    \textbf{1}\,[\!&\begin{array}{l}
         01  \\
         10  \\
    \end{array}\\[-3.5pt]
         \textbf{2}\,\text{\footnotesize[}\! &\begin{array}{c}
          11 
         \end{array}
\end{aligned}\;\;\;\;
\begin{aligned}
    &\;\text{\underline{$C$}}\\[-2pt]
    \textbf{0} \bigg[\hspace{-0.25em}\text{\small-}\!&\begin{array}{l}
         00  \\
         01  \\
         10  \\
    \end{array} \\[-3.5pt] \textbf{1}\text{\small[}\!&\begin{array}{c}
          11 \\
    \end{array}
\end{aligned}\;\;\;\;
\begin{aligned}\;\;\;\;\;\;\;\;
      &\text{\underline{Noise Tolerance}} \\
       &\\[-0.54pt]
       &\;\;\;\;\;\;\;\, \approx0.96\\[-3.54pt]
       &\\[-2.54pt]
       &\\[-1pt]
\end{aligned}
\end{align}
Where the \textbf{boldface} numbers designate the new, merged outputs and $A,B,C$ are Alice, Bob and Charlie respectively, and we use the notation $a'x$ for the outputs of A, $b'y$ for B, and $xy$ for C. In order to obtain the noise tolerances in \eqref{D:4} and \eqref{D:3} both expressible sets \eqref{D:1} and \eqref{D:11} were implemented in the linear program however, in \eqref{D:2} only the expressible set in \eqref{D:1} was used.  Additionally, to obtain the noise tolerance for both of the distributions in \eqref{D:3} non-certificate type constraints were implemented in the linear program. These equality constraints are of the form

\begin{align*}
    P\left(A_{11 \bullet}A_{22 \bullet}B_{2\bullet 1}B_{2\bullet 2}C_{\bullet 21}C_{\bullet 22}\right)=&\underbrace{P\left(A_{11 \bullet}\right)}\underbrace{P\left(A_{22 \bullet}B_{2\bullet 1}B_{2\bullet 2}C_{\bullet 21}C_{\bullet 22}\right)}.\\&\;\,\textit{injectable}\;\;\;\;\;\;\;\;\;\;\,\textit{not injectable}
\end{align*}

In total there are 3 such semi-knowable (semi-injectable) sets with the factorized injectable parts \cite{wolfe2019inflation} being $P(B_{1\bullet 1})$ and $P(C_{\bullet 11})$ for the other two semi-knowable sets respectively. Merging Alice's first two outputs, doing the same for Bob's middle two outputs and Charlie's first three outputs as depicted in the right side of \eqref{D:2} the following inequality was obtained from the linear program

\begin{align*}
    P\left(000\right)\left[P\left(020\right) + P\left(101\right) + P\left(111\right) + P\left(120\right) + P\left(201\right) + P\left(211\right) + P\left(220\right)\right] &+ \\P\left(001\right)\left[P\left(020\right) + P\left(021\right) + P\left(120\right) + P\left(121\right) + P\left(220\right) + P\left(221\right)\right] &+
    \\P\left(101\right)\left[P\left(110\right) + P\left(120\right) + P\left(121\right) + P\left(200\right) + P\left(210\right) + P\left(220\right)\right] &+
    \\P\left(010\right)\left[P\left(101\right) + P\left(201\right)\right] + P\left(020\right)\left[P\left(101\right) + P\left(201\right)\right] &+ \\P\left(021\right)\left[P\left(101\right) + P\left(201\right)\right] + P\left(100\right)\left[P\left(101\right) + P\left(111\right)\right] &+ \\P\left(110\right)\left[P\left(210\right) + P\left(220\right)\right] +
    P\left(201\right)\left[P\left(220\right) + P\left(221\right)\right]&+
    \\P\left(120\right)\left[P\left(200\right) + P\left(201\right) + P\left(210\right) + P\left(220\right)\right] &+
     \\P\left(121\right)\left[P\left(200\right) + P\left(201\right) + P\left(210\right) + P\left(220\right)\right] &+
    \\P\left(111\right)\left[P\left(200\right) + P\left(210\right) + P\left(220\right)\right] &+
    \\P\left(200\right)P\left(220\right) -P\left(100\right)P\left(221\right)&\geq0.
\end{align*}

Where $P\left(ABC\right)=P(ijk)$ with $i,j\in\{0,1,2\}$ and $k\in\{0,1\}$. However, since the distribution at the right side of \eqref{D:2} was obtained by merging outcomes in the original 4 output Fritz distribution, many of the probabilities in this inequality are zero. Therefore, the inequality violated by this distribution was in fact

\begin{align*}
    P(000)\left[P(020)+P(120)+P(211)\right]+P(120)\left[P(200)+P(210)\right]+P(110)P(210)-P(100)P(221)\geq0.
\end{align*}

The violation of this inequality by the noiseless distribution \eqref{D:2} is $\approx -1.6\times10^{-3}$. 

\bibliography{main}

\end{document}